\documentclass[11pt,a4paper]{article}

\usepackage[english]{babel}
\usepackage[T1]{fontenc}
\usepackage[utf8]{inputenc}
\usepackage{fullpage}
\usepackage{complexity} 
\usepackage{lmodern}
\usepackage{amsmath}
\usepackage{amsfonts}
\usepackage{amssymb}
\usepackage{amsthm}
\usepackage{comment}
\usepackage{cite} 
\usepackage{mathrsfs}
\usepackage{blindtext, rotating} 
\usepackage{diagbox} 
\usepackage{multicol} 
\usepackage{changepage} 
\usepackage{enumerate}
\usepackage{wrapfig} 
\usepackage{caption} 
\usepackage{subcaption}
\usepackage{graphicx}
\usepackage{multirow} 
\usepackage{todonotes}
\usepackage{tikz}
\usetikzlibrary{arrows}
\usetikzlibrary{shapes}
\usetikzlibrary{decorations.pathreplacing}
\usetikzlibrary{shapes.geometric}
\usetikzlibrary{patterns}
\usetikzlibrary{fadings}
\tikzstyle{vert2}=[circle,inner sep=1.5,fill=white,draw,minimum size=.2cm]
\tikzstyle{vert3}=[inner sep=1.5,fill=white,draw=black,minimum size=.2cm]
\usetikzlibrary{calc}
\usetikzlibrary{arrows.meta}

\usepackage[hidelinks]{hyperref}
\usepackage[capitalise,nameinlink, noabbrev]{cleveref}
\Crefname{ALC@unique}{Line}{Lines} 

\newtheorem{theorem}{Theorem}
\newtheorem{lemma}[theorem]{Lemma}
\newtheorem{observation}[theorem]{Observation}
\newtheorem{corollary}[theorem]{Corollary}
\crefname{corollary}{Corollary}{Corollaries}

\newtheorem{claim}[theorem]{Claim}
\theoremstyle{definition}
\newtheorem{definition}[theorem]{Definition}

\newcommand{\ie}{i.\,e.,\ }

\newcommand{\etal}{et.\,al.\ }

\usepackage{tabularx}

\newcommand{\dTVC}{$\Delta$-TVC}

\graphicspath{{Figures/}}

\title{\vspace{-0.5cm} The Complexity of Temporal Vertex Cover in Small-Degree Graphs}

\author{Thekla Hamm\thanks{%
		Algorithms and Complexity Group, TU Wien, Vienna, Austria. Email: \texttt{thamm@ac.tuwien.ac.at}}
	\and Nina Klobas\thanks{%
		Department of Computer Science, Durham University, UK. Email: \texttt{nina.klobas@durham.ac.uk}}
	\and George B.~Mertzios \thanks{%
		Department of Computer Science, Durham University, UK. Email: \texttt{george.mertzios@durham.ac.uk}} \thanks{Supported by the EPSRC grant EP/P020372/1.} 
	\and Paul G. Spirakis\thanks{%
		Department of Computer Science, University of Liverpool, UK and Computer Engineering \& Informatics Department, University of Patras, Greece. Email: \texttt{p.spirakis@liverpool.ac.uk}} \thanks{Supported by the NeST initiative of the School of EEE and CS at the University of Liverpool and by the EPSRC grant EP/P02002X/1.}}
\date{\vspace{-1.0cm}}

\newif\iflong
\newif\ifshort

\longtrue

\iflong
\else
\shorttrue
\fi

\begin{document}

\maketitle

\begin{abstract}
	Temporal graphs naturally model graphs whose underlying topology changes over time. Recently, the problems \textsc{Temporal Vertex Cover} (or \textsc{TVC}) and \textsc{Sliding-Window Temporal Vertex Cover} (or \textsc{$\Delta$-TVC} for time-windows of a fixed-length $\Delta$) have been established as natural extensions of the classic problem \textsc{Vertex Cover} on static graphs with connections to areas such as surveillance in sensor networks.
	
	In this paper, we initiate a systematic study of the complexity of \textsc{TVC} and \textsc{$\Delta$-TVC} on sparse graphs. Our main result shows that for every $\Delta\geq 2$, \textsc{$\Delta$-TVC} is NP-hard even when the underlying topology is described by a path or a cycle.
	This resolves an open problem from literature and shows a surprising contrast between \textsc{\(\Delta\)-TVC} and \textsc{TVC} for which we provide a polynomial-time algorithm in the same setting.
	To circumvent this hardness, we present some fixed-parameter and approximation algorithms. 
\end{abstract}

\section{Introduction}

A great variety of modern, as well as traditional networks, are dynamic in nature 
as their link availability changes over time. 
Information and communication networks, social networks, transportation networks, 
and various physical systems are only a few indicative examples of such inherently dynamic networks~\cite{Holme-Saramaki-book-13,michailCACM}. 
All these application areas share the common characteristic that the network structure, 
i.e.~the underlying graph topology, is subject to \emph{discrete changes over time}. 
In this paper, embarking from the foundational work of Kempe et al.~\cite{kempe}, 
we adopt the following simple and natural model for time-varying networks, 
which is given by a graph with sets of time-labels associated with its edges, while the vertex set is fixed. 

\begin{definition}[Temporal Graph]
	\label{def:temporal-graph} A \emph{temporal graph} is a pair $(G,\lambda)$,
	where $G=(V,E)$ is an underlying (static) graph and $\lambda :E\rightarrow
	2^{\mathbb{N}}$ is a \emph{time-labeling} function which assigns to every
	edge of $G$ a set of discrete-time labels.
\end{definition}

For an edge~$e \in E$ in the underlying graph~$G$ of a temporal graph~$(G,\lambda)$, 
if $t \in \lambda(e)$ then we say that $e$ is \emph{active} at time~$t$ in $(G,\lambda)$. 
That is, for every edge $e\in E$, $\lambda (e)$ denotes the set of time slots at which $e$ is \emph{active}. 
Due to their relevance and applicability in many areas, 
temporal graphs have been studied from various perspectives and under different names 
such as \emph{dynamic}~\cite{GiakkoupisSS14,CasteigtsFloccini12}, 
\emph{evolving}~\cite{xuan,Ferreira-MANETS-04,clementi10}, 
\emph{time-varying}~\cite{FlocchiniMS09,TangMML10-ACM,krizanc1}, 
and \emph{graphs over time} \cite{Leskovec-Kleinberg-Faloutsos07}. 
For a comprehensive overview of the existing models 
and results on temporal graphs from a distributed computing perspective see the 
surveys~\cite{CasteigtsFloccini12,flocchini1,flocchini2}.

Mainly motivated by the fact that, due to causality, information can be 
transferred in a temporal graph along sequences of edges whose time-labels 
are increasing, the most traditional research on temporal graphs has focused 
on temporal paths and other ``path-related'' notions, 
such as e.g.~temporal analogues of distance, reachability, exploration and centrality~\cite{KlobasMMNZ21,HeegerHMMNS21,akrida16,Erlebach0K21,MertziosMS19,michailTSP16,akridaTOCS17,EnrightMMZ21,ZschocheFMN20,CasteigtsHMZ21}. 
To complement this direction, several attempts have been recently made to 
define meaningful ``non-path'' temporal graph problems which appropriately 
model-specific applications. 
Motivated by the contact patterns among high-school students, 
Viard et al.~\cite{viardCliqueTCS16}, 
introduced $\Delta$-cliques, an extension of the concept of cliques to temporal graphs (see also~\cite{himmel17,BHMMNS18}). 
Chen et al.~\cite{ChenMSS18} presented an extension of the cluster editing problem to temporal graphs, 
 in which all vertices interact with each other at least once every $\Delta$ consecutive time-steps within a given time interval.
Furthermore, Akrida et al.~\cite{AkridaMSZ20} introduced the notion of temporal vertex cover (also with a sliding time window), 
motivated by applications of covering problems in sensor networks. 
Further examples of meaningful ``non-path'' temporal graph problems include variations of temporal graph coloring~\cite{MertziosMZ21,yu2013algorithms,ghosal2015channel} in the context of planning and channel assignment in mobile sensor networks, 
and the temporally transitive orientations of temporal graphs~\cite{MertziosMRSZ21}.


The problems \textsc{Temporal Vertex Cover} (or \textsc{TVC}) 
and \textsc{Sliding-Window Temporal Vertex Cover} (or \textsc{$\Delta$-TVC} for time-windows of a fixed-length $\Delta$) have been established 
as natural extensions of the well-known \textsc{Vertex Cover} problem on static graphs~\cite{AkridaMSZ20}. 
Given a temporal graph $\mathcal{G}$, the aim of \textsc{TVC} is to cover every edge at least once during the lifetime $T$ of $\mathcal{G}$, 
where an edge can be covered by one of its endpoints, and only at a time-step when it is active. 
For any $\Delta\in \mathbb{N}$, the aim of the more ``pragmatic'' problem \textsc{$\Delta$-TVC} is to cover every edge at least once at every $\Delta$ consecutive time-steps.
In both problems, we try to minimize the number of ``vertex appearances'' in the resulting cover, where a vertex appearance is a pair $(v,t)$ 
for some vertex $v$ and $t\in \{1,2,\ldots,T\}$.

\textsc{TVC} and \textsc{$\Delta$-TVC} naturally generalize the applications of the static problem \textsc{Vertex Cover} to more dynamic inputs, especially in the areas of wireless ad hoc networks, as well as network security and scheduling. 
In the case of a static graph, the vertex cover can contain trusted vertices that have the ability to monitor/surveil all transmissions~\cite{ileri2016vertex,RichterHG07} or all link failures~\cite{KavalciUD14} between any pair of vertices through the edges of the graph. In the temporal setting, it makes sense to monitor the transmissions and to check for link failures within every sliding time window of an appropriate length $\Delta$ (which is exactly modeled by $\Delta$-TVC).

It is already known that both \textsc{TVC} and \textsc{$\Delta$-TVC} are NP-hard; for \textsc{$\Delta$-TVC} 
this is even the case when $\Delta=2$ and the minimum degree of the 
underlying graph $G$ is just 3~\cite{AkridaMSZ20}. One of the most intriguing questions left open (see Problem 1 in~\cite{AkridaMSZ20}) is 
whether \textsc{$\Delta$-TVC} (or, more generally, \textsc{Sliding-Window Temporal Vertex Cover}) can be solved in polynomial time 
on always degree at most $2$ temporal graphs, that is, on temporal graphs where the maximum degree of the graph at each time-step is at most 2.

\medskip
\noindent\textbf{Our Contribution.} 
In this paper, we initiate the study of the complexity of \textsc{TVC} and \textsc{$\Delta$-TVC} on sparse graphs. 
Our main result (see~\cref{subsec:2tvc-paths-NP}) is that, for every $\Delta\geq 2$, \textsc{$\Delta$-TVC} is NP-hard even when $G$ is a path or a cycle, while \textsc{TVC} can be 
solved in polynomial time on paths and cycles. This resolves the first open question (Problem 1) of~\cite{AkridaMSZ20}. 
In contrast, we show that \textsc{TVC} (see~\cref{subsec:algres}) can be solved in polynomial time on temporal paths and cycles. 
Moreover, for any $\Delta\geq 2$, we provide a \emph{Polynomial-Time Approximation Scheme (PTAS)} for $\Delta$-TVC on temporal paths and cycles (see~\cref{subsec:algres}), which also complements our hardness result for paths. 

The \NP-hardness of \Cref{subsec:2tvc-paths-NP} signifies that an optimum solution for \(\Delta\)-TVC is hard to compute, even for $\Delta=2$ and under severe degree restrictions of the input instance. 
To counter this hardness for more general temporal graphs than those with underlying paths and cycles as in \Cref{paths-cycles-sec}, in \Cref{bounded-degree-sec} we give three algorithms for every $\Delta\geq 2$:
First we present an exact algorithm for \(\Delta\)-TVC with 
exponential running time dependency on the number of edges in the underlying graph
(see \Cref{subsec:exact}).
Using this algorithm we are able to devise, for any $d \geq 3$, a polynomial-time \((d-1)\)-approximation (see \Cref{better-d-approx-subsec}), where \(d\) is the maximum vertex degree in any time-step, \ie in any part of the temporal graph that is active at the same time.
This improves the currently best known \(d\)-approximation algorithm for \textsc{\(\Delta\)-TVC}~\cite{AkridaMSZ20} and thus also answers another open question (Problem 2 in~\cite{AkridaMSZ20}).
Finally, we present a simple fixed-parameter tractable algorithm with respect to the size of an optimum solution (see \Cref{cp-alg-subsec}).

\section{Preliminaries}\label{prelim-sec}

Given a (static) graph $G=(V,E)$ with vertices in $V$ and edges in $E$, an edge between two vertices~$u$ and~$v$ is denoted by $uv$, 
and in this case~$u$ and~$v$ are said to be \emph{adjacent} in~$G$. 
For every $i,j\in \mathbb{N}$, where $i\leq j$, we let $[i,j]=\{i,i+1,\ldots ,j\}$ and $[j]=[1,j]$. 
Throughout the paper, we consider temporal graphs whose underlying graphs are finite and whose time-labeling functions only map to finite sets.
This implies that there is some \(t \in \mathbb{N}\) such that, for every \(t' > t\), no edge of \(G\) is active at \(t'\) in \((G, \lambda)\).
We denote the smallest such \(t\) by $T$, \ie $T = \max\{t \in \lambda(e) \mid e \in E\}$, and call $T$ the \emph{lifetime} of \((G, \lambda)\). 
Unless otherwise specified, $n$ denotes the number of vertices in the underlying graph $G$, and $T$ denotes the lifetime of the temporal graph $\mathcal{G}$. 
We refer to each
integer $t\in [T]$ as a \emph{time slot} of $(G,\lambda )$.
The \emph{instance} (or \emph{snapshot}) of $(G,\lambda )$ 
\emph{at time}~$t$ is the static graph $G_{t}=(V,E_{t})$, where $%
E_{t}=\{e\in E:t\in \lambda (e)\}$. 

A \emph{temporal path} of length $k$ is a temporal graph $\mathcal{P}=(P, \lambda)$,
where the underlying graph $P$ is the path $(v_0, v_1, v_2, \dots, v_k)$ on $k+1$ vertices,
with edges $e_i = v_{i-1}v_i$ for $i = 1, 2, \dots, k$.
In many places throughout the paper, we visualize a temporal path as a 2-dimensional array  
$V(P) \times [T]$, where two vertices $(x,t), (y,t') \in V(P) \times [T]$ are connected in this array if and only if $t = t'\in \lambda(xy)$ and $xy \in E(P)$.
For example see \cref{fig:ExampleTemporalPath}.

\begin{figure}[h]
	\centering
	\includegraphics[width=0.75\linewidth]{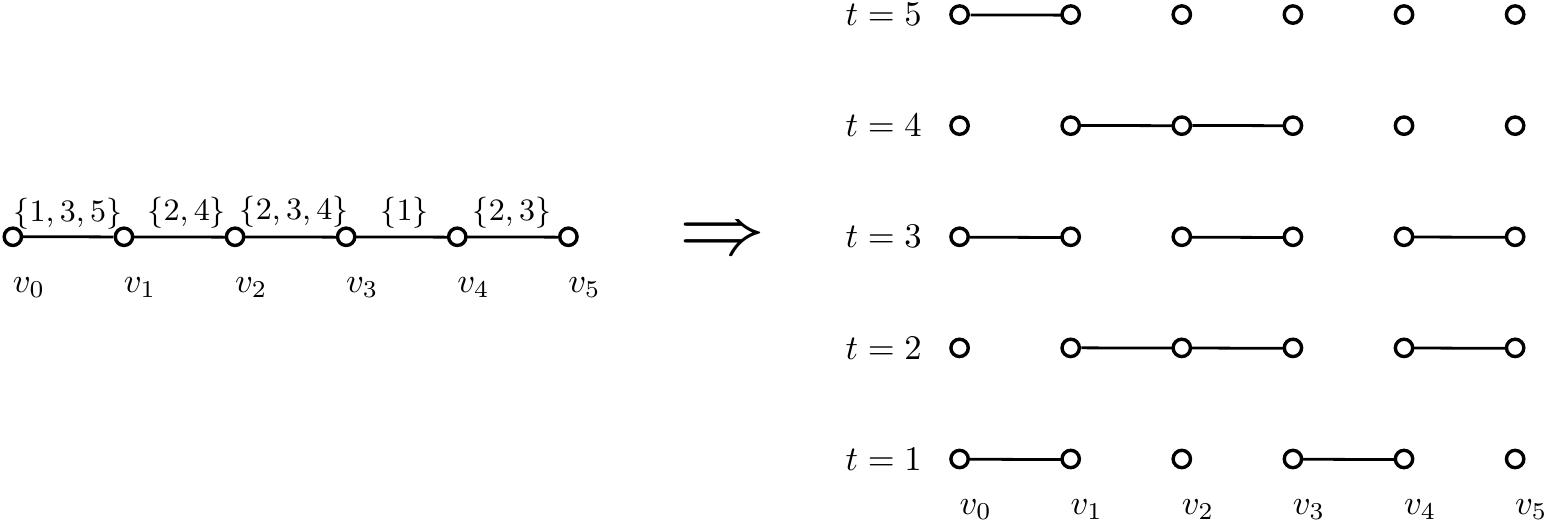}
	\caption{An example of visualizing a temporal path graph $\mathcal{G}$ as a $2$-dimensional array, in which every edge corresponds to a time-edge of $\mathcal{G}$.}
	\label{fig:ExampleTemporalPath} 
\end{figure}

For every $t=1,\ldots,T-\Delta+1$, let $W_t = [t, t+ \Delta -1]$ be the $\Delta$-time window that starts at time $t$.
For every $v\in V$ and every time slot $t$, we denote the \emph{appearance
of vertex} $v$ \emph{at time} $t$ by the pair $(v,t)$ and the
\emph{edge appearance} (or \emph{time-edge}) of $e$ at time $t$ by $(e,t)$.

A \emph{temporal vertex subset} of \((G, \lambda)\) is a set of vertex appearances in $(G, \lambda)$, i.e.\ a set of the form \(S \subseteq \{(v, t) \mid v \in V, t \in [T]\}\).
For a temporal vertex subset~\(S\) and some \(\Delta\)-time window~\(W_i\) within the lifetime of \((G, \lambda)\), we denote by \(S[W_i] = \{(v, t) \in S \mid t \in W_i\}\) the subset of all vertex appearances in \(S\) in the \(\Delta\)-time window~\(W_i\).
For a \(\Delta\)-time window~\(W_i\) within the lifetime of a temporal graph~\((G, \lambda)\), we denote by \(E[W_i] = \{e \in E \mid \lambda(e) \cap W_i \neq \emptyset\}\) the set of all edges which appear at some time slot within \(W_i\).

A temporal vertex subset~$\mathcal{C}$ is a \emph{sliding \(\Delta\)-time window temporal vertex cover}, or \(\Delta\)-TVC for short, 
of a temporal graph \((G, \lambda)\) 
if, for every \(\Delta\)-time window~\(W_i\) within the lifetime of \((G, \lambda)\) and for every edge in \(E[W_i]\), there is some \((v,t) \in \mathcal{C}[W_i]\) such that \(v \in e\), i.e.\ \(v\) is an endpoint of \(e\), and \(t \in \lambda(e)\).
Here we also say \((v,t)\) \emph{covers} \((e,t)\) in time window \(W_i\).

\section{Paths and Cycles}\label{paths-cycles-sec}
In Section~\ref{subsec:2tvc-paths-NP} we provide our main NP-hardness result for $\Delta$-TVC, for any $\Delta\geq 2$, 
on instances whose underlying graph is a path or a cycle (see Theorem~\ref{thm:2tvc-paths-NP} and 
Corollary~\ref{cor:2tvc-cycles-NP}). 
In Section~\ref{subsec:algres} we prove that TVC on underlying paths and cycles is polynomially solvable, and we also provide our PTAS for $\Delta$-TVC on underlying paths and cycles, for every $\Delta\geq 2$.

\subsection{Hardness on Temporal Paths and Cycles}\label{subsec:2tvc-paths-NP}
 
Our NP-hardness reduction of \cref{thm:2tvc-paths-NP} is done from the NP-hard problem \emph{planar monotone rectilinear $3$ satisfiability} 
(or \emph{planar monotone rectilinear 3SAT}), see \cite{Berg2010Optimal}.
This is a specialization of the classical Boolean $3$-satisfiability problem to a planar incidence graph. 
A Boolean satisfiability formula $\phi$ in conjunctive normal form (CNF) is called \emph{monotone} if each clause of $\phi$ consists of only positive or only negative literals. We refer to these clauses as \emph{positive} and \emph{negative} clauses, respectively. 
By possibly repeating literals, we may assume without loss of generality 
that every clause contains exactly three (not necessarily different) literals.

In an instance of planar monotone rectilinear $3$SAT, each variable and each clause is represented with a horizontal line segment, as follows. 
The line segments of all variables lie on the same horizontal line on the plane, which we call the \emph{variable-axis}.
For every clause $C = (x_i \lor x_j \lor x_k)$ (or $C = (\overline{x_i} \lor \overline{x_j} \lor \overline{x_k})$), the line segment of $C$ is connected via straight vertical line segments to the line segments of $x_i, x_j$ and of $x_k$, such that every two (horizontal or vertical) line segments are pairwise non-intersecting. 
Furthermore, every line segment of a positive (resp.~negative) clause 
lies above (resp.~below) the variable-axis. 
Finally, by possibly slightly moving the clause line segments higher or lower, we can assume without loss of generality that every clause line segment lies on a different horizontal line on the plane. 
\iflong
For an example see \cref{fig:rectilinearMax2SAT}.

\begin{figure}[h]
	\centering
	\includegraphics[width=0.6\linewidth]{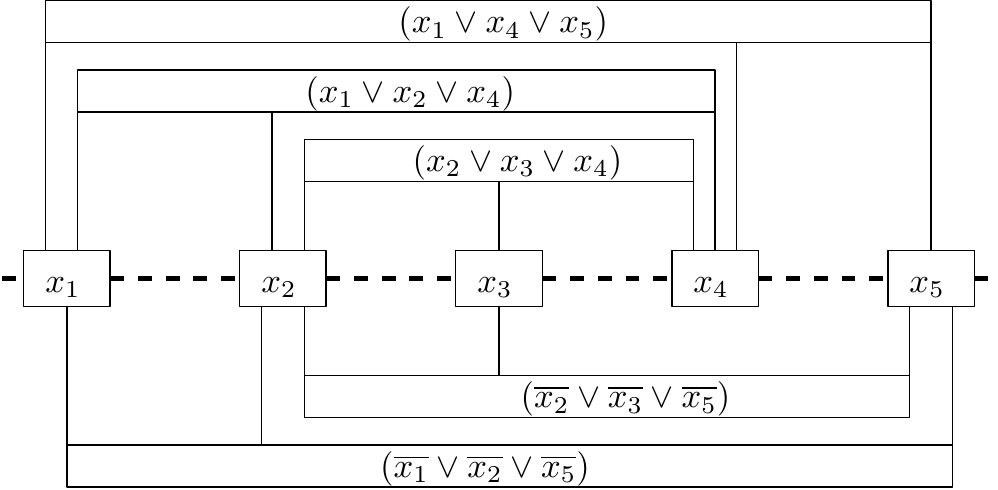}
	\caption{An example of an instance of a planar monotone rectilinear 3SAT 
		$\phi = (x_2 \lor x_3 \lor x_4) \land (x_1 \lor x_2 \lor x_4) \land (x_1 \lor x_4 \lor x_5) \land (\overline{x_2} \lor \overline{x_3} \lor \overline{x_5}) \land (\overline{x_1} \lor \overline{x_2} \lor \overline{x_5})$. For visual purposes, the line segments for variables and for clauses are illustrated here with boxes.}
	\label{fig:rectilinearMax2SAT}
\end{figure}
\fi

Let $\phi$ be an arbitrary instance of planar monotone rectilinear $3$SAT, 
where $X = \{x_1, \dotsc, x_{n}\}$ is its set of Boolean variables and $\phi(X) = \{C_1, \dotsc, C_{m}\}$ is its set of clauses.
We construct from $\phi$ a temporal path $\mathcal{G_\phi}$ 
and prove \iflong (see~\cref{lem:two-directions}) \fi that there exists a truth assignment of $X$ which satisfies $\phi(X)$ 
if and only if the optimum value of $2$-TVC on $\mathcal{G_\phi}$ is at most $f(\mathcal{G_\phi})$. 
The exact value of $f(\mathcal{G_\phi})$ will be defined later.

\subsubsection{High-Level Description}
Given a representation (i.e.~embedding) $R_{\phi}$ of an instance $\phi$ of planar monotone rectilinear $3$SAT, 
we construct a $2$-dimensional array of the temporal path $\mathcal{G}_{\phi}$, where:
\begin{itemize}
	\item every variable (horizontal) line segment in $R_{\phi}$ is associated with one or more \emph{segment blocks} (to be formally defined below) in $\mathcal{G}_{\phi}$, and 
	\item every clause (horizontal) line segment in $R_{\phi}$, corresponding to the clause $C=(x_i \lor x_j \lor x_k)$ (resp.~$C=(\overline{x_i} \lor \overline{x_j} \lor \overline{x_k})$), 
	is associated with a clause gadget in $\mathcal{G}_{\phi}$, which
	consists of three edges (one for each of $x_i, x_j, x_k$), each appearing in $4$ consecutive time-steps, 
	together with two paths connecting them in the 2-dimensional array for $\mathcal{G}$ (we call these paths the \emph{clause gadget connectors}, for an illustration see~\cref{fig:NP-rectilinearConstruction}),
	\item every vertical line segment in $R_{\phi}$, connecting variable line segments to clause line segments, 
	is associated with an edge of $\mathcal{G}_{\phi}$ that appears in consecutive time-steps.
\end{itemize}

The exact description of the variables' and clauses' gadgets is given below; first, we need to precisely define the segment blocks.

\paragraph*{Segment blocks} 
are used to represent variables.
Every segment block consists of a path of length $7$ on vertices $(u_0, u_1, \dots , u_7)$,
where the first and last edges (\ie $u_0u_1$ and $u_6 u_7$) appear at 
$9$ consecutive time-steps starting at time $t$ and ending at time $t+8$,
with all other edges appearing only two times, \ie at times $t+1$ and $t + 7$.
\iflong For an example see \cref{fig:NP-varibleGadgets}.

\begin{figure}[h]
	\centering
	\includegraphics[width=0.3\linewidth]{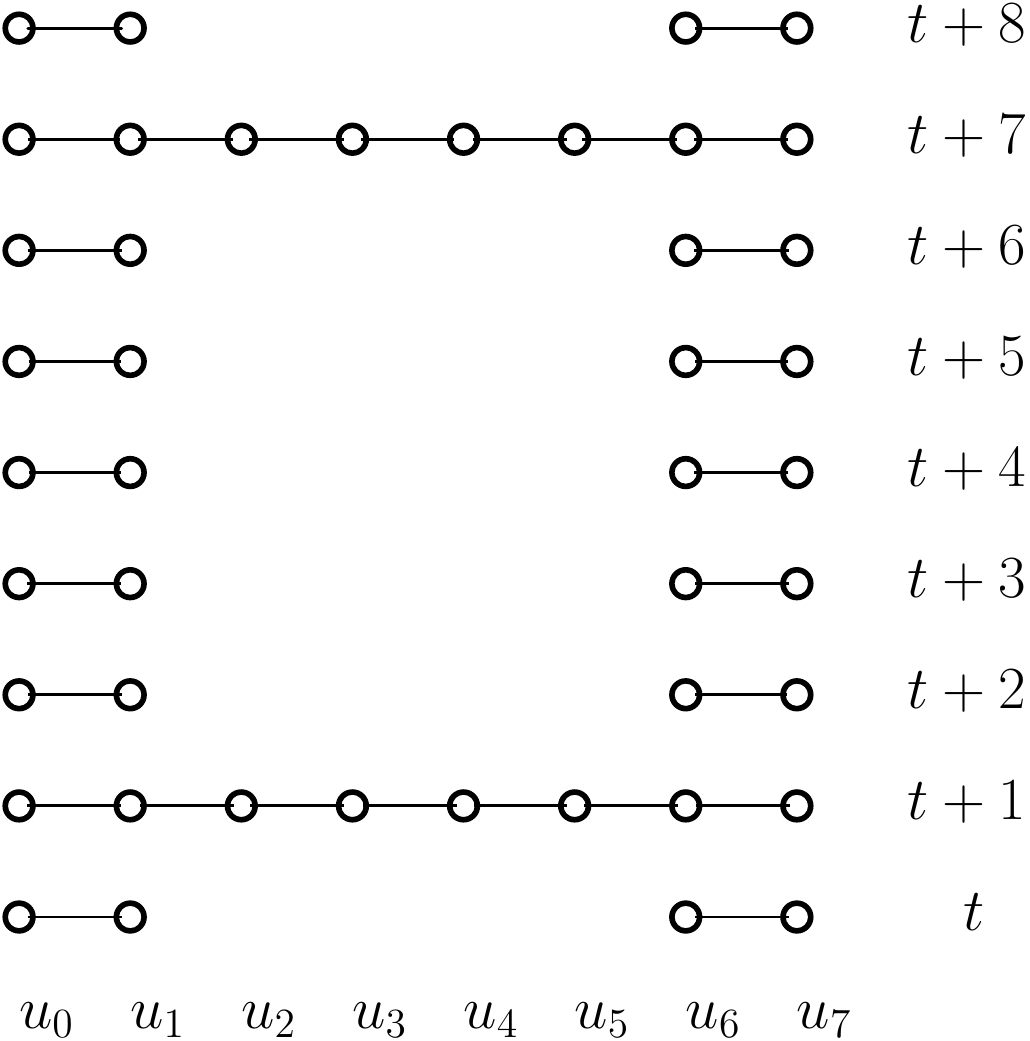}
	\caption{Example of a segment block construction.}
	\label{fig:NP-varibleGadgets}
\end{figure}\fi

Time-edges which correspond to the first and last appearances of $u_0u_1$ and $u_6u_7$ in a segment block are called \emph{dummy time-edges},
all remaining time-edges 
form two (bottom and top) horizontal paths, and 
two (left and right) vertical sequences of time-edges (which we call here \emph{vertical paths}), 
\iflong see \cref{fig:NP-variableGadgets-names}\fi.
\iflong
\begin{figure}[h]
	\centering
	\includegraphics[width=0.3\linewidth]{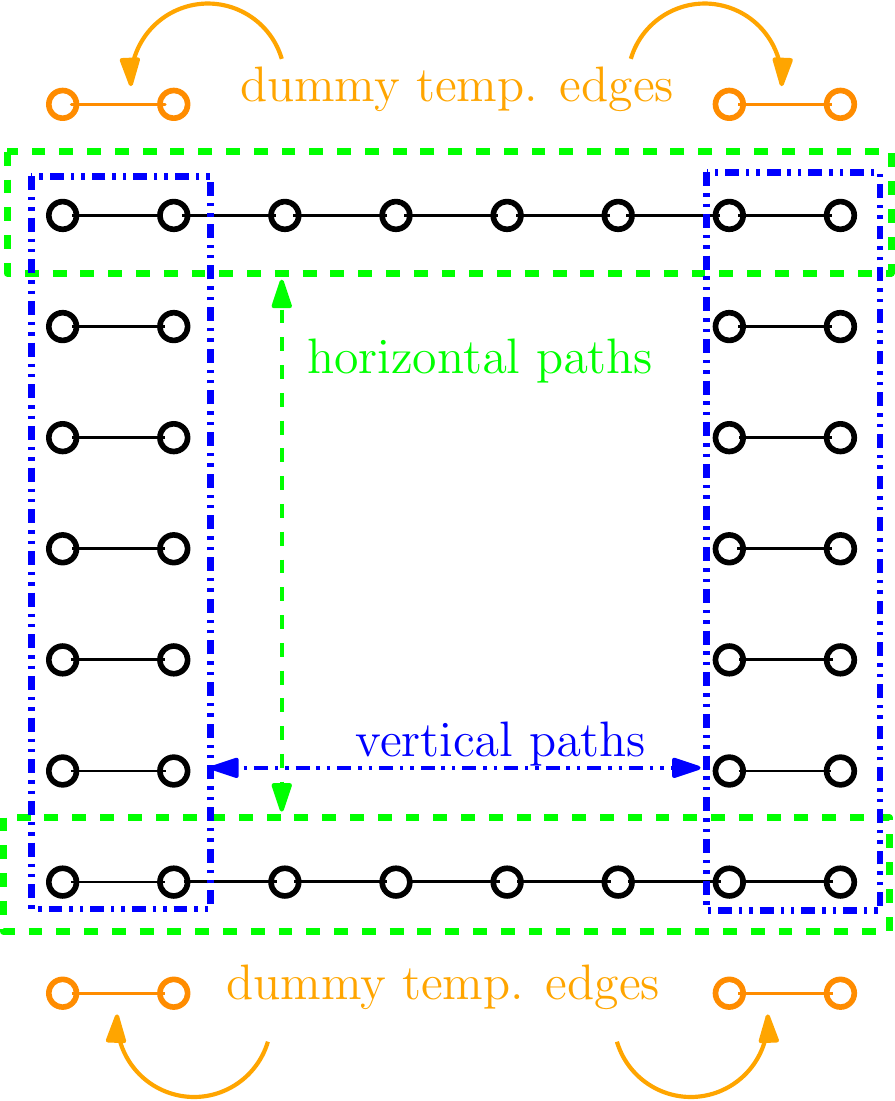}
	\caption{An example where dummy time-edges, and vertical and horizontal paths are depicted.}
	\label{fig:NP-variableGadgets-names}
\end{figure}\fi
Using the next technical lemma will allow us to model the two different truth values of each variable $x_i$ (True, resp.~False) via two different optimum solutions of 2-TVC on a segment block (namely the ``orange and green'', resp.~``orange and red'' temporal vertex covers of the segment block, see Figure~\ref{fig:NP-varibleGadgetsCover}).

\begin{lemma}\label{lem:NP-CoverVariableGadgets}
	There are exactly two different optimum solutions for $2$-TVC of a segment block,
	both of size $15$.
\end{lemma}

\begin{proof}
	Let $\mathcal{C}$ be a $2$-TVC of a segment block on vertices $u_0, \dots, u_7$ that starts at time $t$ and finishes at time $t + 8$.
	
	In order to cover the dummy time-edges in time windows $W_{t-1}$ and $W_{t+8}$ one of their endpoints has to be in $\mathcal{C}$.
	Now let us start with the covering of the first edge ($u_0u_1$) at time $t+1$.
	Since the dummy time-edges are covered, the edge $u_0u_1$ is covered in the time window $W_t$ 
	but it is not yet covered in the time window $W_{t+1}$. 
	We have two options: cover it at time $t+1$ or $t+2$.
	\begin{itemize}
		\item 
		Suppose that we cover the edge $u_0u_1$ at $t+1$, then the next time-step it has to be covered is $t+3$, the next one $t+5$, and the last one $t+7$.
		Now that the left vertical path is covered we proceed to cover the bottom and top horizontal paths.
		The middle $5$ edges, from $u_1$ to $u_6$, appear only at time-steps $t+1$ and $t + 7$.
		Since we covered the edge $u_0u_1$ at time $t+1$, we can argue that the optimum solution includes the vertex appearance $(u_1, t+1)$ and therefore the edge $u_1u_2$ is also covered.
		Extending this covering optimally to the whole path we need to add every second vertex to $\mathcal{C}$, \ie vertex appearances $(u_3,t+1)$ and $(u_5, t+1)$.
		Similarly, the same holds for the appearances of vertices $u_1, \dots , u_6$ at time $t + 7$.
		The last thing we need to cover is the right vertical path.
		Since the edge $u_6u_7$ is covered at time $t$, the next time-step we have to cover it is $t+2$, which forces the next cover to be at $t+4$ and the last one at $t+6$.
		
		In total $\mathcal{C}$ consists of $4$ endpoints of the dummy time-edges, $4$ vertices of the left and $3$ of the right vertical paths, $2$ vertices of the bottom and $2$ of the top horizontal paths.
		Altogether we used $11$ vertices to cover vertical and horizontal paths
		and $4$ for dummy time-edges.
		The above described $2$-TVC corresponds to the ``orange and red'' vertex appearances of the segment block depicted in the \cref{fig:NP-varibleGadgetsCover}.
		
		Let us also emphasize that, with the exception of times $t+1$ and $t + 7$, we do not distinguish between the solutions that use a different endpoint to cover the first and last edge.
		For example, if a solution covers the edge $u_0 u_1$ at time $t+2$ then we do not care which of $(u_0,t+2)$ or $(u_1,t+2)$ is in the TVC.
		
		\item
		Covering the edge $u_0u_1$ at time $t+2$ produces the $2$-TVC that is a mirror version of the previous one on the vertical and horizontal paths.
		More precisely, in this case 
		the covering consists of 
		$3$ vertices of the left and $4$ of the right vertical paths and again $2$ vertices of the bottom and $2$ of the top horizontal paths,
		together with $4$ vertices covering the dummy time-edges.
		
		This $2$-TVC corresponds to the ``orange and green'' vertex appearances of the segment block depicted in the \cref{fig:NP-varibleGadgetsCover}.
	\end{itemize}
	If we start with vertex appearances from one solution and add vertex appearances from the other solution
	then we
	either create a $2$-TVC of bigger size or leave some edges uncovered.
	Therefore, the optimum temporal vertex cover of any segment block consists of only ``orange and red'' or ``orange and green'' vertex appearances.
\end{proof}
	
	\begin{figure}[h]
		\centering
		\includegraphics[width=0.3\linewidth]{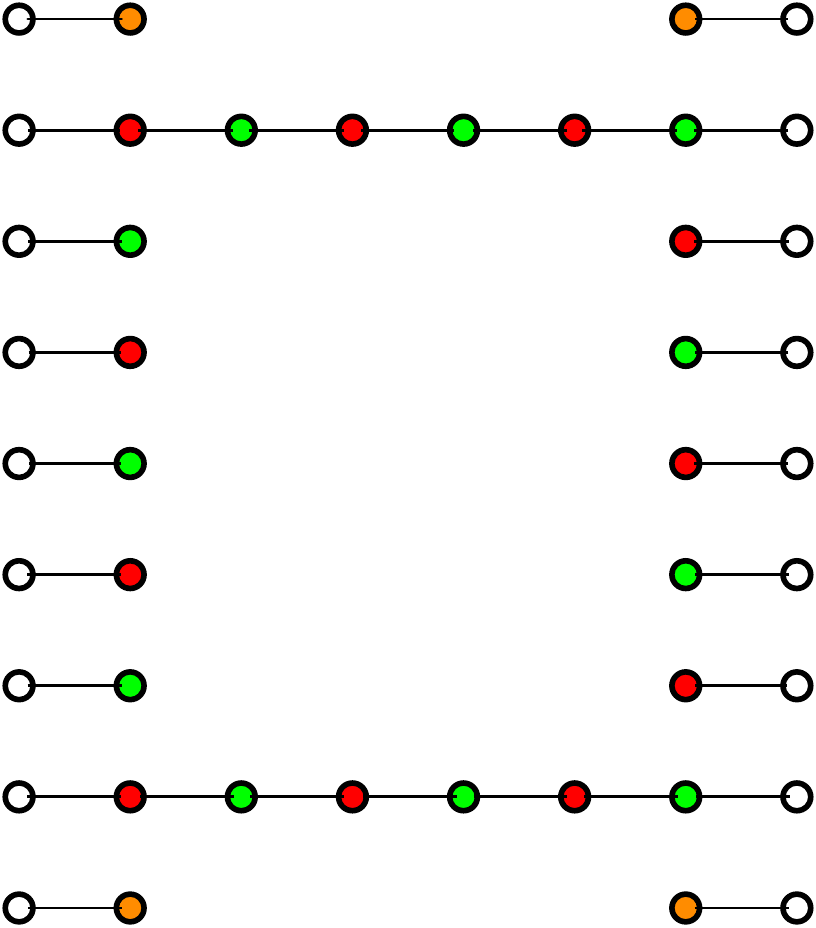}
		\caption{An example of two optimum covers of a segment block: (i)~with the ``orange and green''-colored, or 
		(ii)~with the ``orange and red''-colored vertex appearances.}
		\label{fig:NP-varibleGadgetsCover}
	\end{figure}

For each variable $x_i$ we create multiple copies of segment blocks, 
and some specific pairs of these segment blocks are connected to each other 
via the so-called ``horizontal bridges''. 
Two segment blocks, which are connected via a horizontal bridge, start at the same time $t$ but are built on different sets of vertices (\ie one is to the left of the other in the 2-dimensional array).
All the copies have to be created in such a way, that their optimum $2$-TVCs depend on each other.
In the following, we describe how to connect two different segment blocks (both for the same variable $x_i$). 
As we prove \iflong below (see~\cref{lem:SizeOfTVCinHorizontalDirection})\fi, 
there are two ways to optimally cover this construction: 
one using the ``orange and green'', and one using the ``orange and red'' vertex appearances (thus modelling the truth values True and False of variable $x_i$ in our reduction)\iflong, see~\cref{fig:NP-variableSegments}\fi. 

\iflong
\paragraph{Consistency in the horizontal direction.}
Let $\mathcal{H}$ be a temporal path on vertices 
$(u^1_0, u^1_1, \dots ,$ $u^1_7,w_1,w_2,w_3,w_4 , u^2_0, u^2_1, , \dots , u^2_7)$,
with
two segment blocks on vertices $(u^1_0, u^1_1, \dots , u^1_7)$ 
and $(u^2_0, u^2_1, \dots , u^2_7)$, respectively, 
starting at time $t$ and finishing at time $t+8$,
and let $P = (u^1_6, u^1_7, w_1,w_2,w_3,w_4 , u^2_0, u^2_1)$ be a path of length $7$, that appears exactly at times $t + 2$ and $t + 5$. 
For an example see \cref{fig:NP-variableSegments}. 
We refer to this temporal path $P$ as a \emph{horizontal bridge} between the two segment blocks.

\begin{figure}[h]
		\centering
		\includegraphics[width=0.6\linewidth]{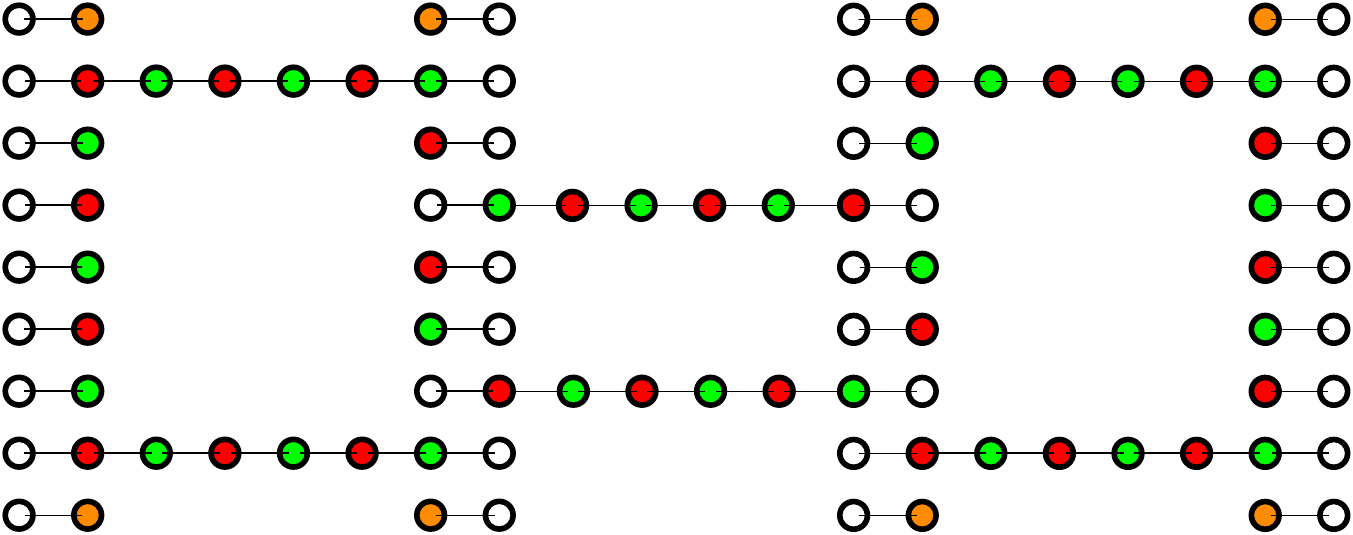}
\caption{An example of connecting two segment blocks in the horizontal direction. 
The two different ways to optimally cover this construction are 
(i)~with the green-colored and the orange-colored, or 
(ii)~with the red-colored and the orange-colored vertex appearances.}
	\label{fig:NP-variableSegments} 
\end{figure}

\begin{lemma}\label{lem:SizeOfTVCinHorizontalDirection}
	The temporal graph $\mathcal{H}$ has exactly two optimum $2$-TVCs, both of size $34$.
\end{lemma}

\begin{proof}
	Let $\mathcal{H}$ consist of two segment blocks.
	Suppose that the first segment block, \ie left one, 
	is covered with the ``orange and green'' $2$-TVC. 
	Then the first edge $u^1_6 u^1_7$ of the path $P$
	does not have to be covered at time $t+2$
	as it is covered at times $t+1$ and $t+3$. 
	For the other $6$ edges of $P$, there exists a unique optimum $2$-TVC, that uses every second vertex at time $t+2$ and is of size $3$.
	Therefore we cover the edge $u^2_0u^2_1$ at time $t+2$ which enforces the ``orange and green'' $2$-TVC
	also on the right segment block.
	Now, the only remaining uncovered edges are four consecutive edges of $P$, from $w_1$ to $u_0^2$ at time $t+5$, which are optimally covered with two vertex appearances.
	Altogether we used $15 + 15$ vertex appearances, to cover both segment blocks and $2 + 2$ to cover $P$ at $t+2, t+5$.
\end{proof}

\cref{lem:SizeOfTVCinHorizontalDirection} ensures that, in an optimum $2$-TVC of the construction of~\cref{fig:NP-variableSegments}, 
either both the left and the right segment blocks contain the ``orange and green'' vertex appearances, or they both contain the ``orange and red'' vertex appearances. 
We can extend the above result to $d$ consecutive copies of the same segment block and get the following result.

\begin{corollary}\label{cor:NP-SizeOfTVCinHorizontalDirection}
	The temporal graph corresponding to $d$ copies of a segment block, where two consecutive are connected 
	via horizontal bridges,
	has exactly two optimum $2$-TVCs of size $19d -4$.
\end{corollary}
\begin{proof}
	The optimum $2$-TVC of each segment block consists of $15$ vertex appearances.
	A horizontal bridge is optimally covered using $4$ extra vertex appearances.
	Therefore all together we have $15 d + 4 (d-1) = 19d -4$ vertex appearances in the optimum $2$-TVC.
\end{proof}

\subsubsection*{Variable Gadget \label{sec:NP-paths-VariableSegments}}
From the planar, rectilinear embedding $R_\phi$ of $\phi$, we can easily fix the order of variables.
We fix the variables in the order they appear in the variable-axis,
starting from the left to the right. 

Let $d_i$ be the number of appearances of variable $x_i$ as a literal (\ie as $x_i$ and $\overline{x_i}$) in $\phi$.
For every variable $x_i$ we create $d_i$ copies of the segment block, 
which follow each other on a horizontal line and are connected via \emph{horizontal bridges}.
Between variable gadgets of two consecutive variables $x_i$ and $x_{i+1}$
we add $4$ vertices (without any additional time-edges).
All variable gadgets in $\mathcal{G_\phi}$ start and finish at the same time
\ie they lie on the same vertical line.

\begin{lemma}\label{lem:NPoddDistance-RightmostEdges}
	The distance between two rightmost (or two leftmost) edges of any pair of segment blocks in $\mathcal{G_\phi}$ is odd.
\end{lemma}
\begin{proof}
	Let $(u_6 u_7, t)$ and $(v_6 v_7,t)$ be the rightmost time-edges of two segment blocks $\mathcal{X}$ and $\mathcal{Y}$ in $\mathcal{G_\phi}$.
	Without loss of generality, we may assume that $\mathcal{X}$ appears before/left of the $\mathcal{Y}$ in $\mathcal{G_\phi}$.
	
	If $\mathcal{X}$ and $\mathcal{Y}$ are consecutive (\ie right next to each other) 
	then there are $4$ vertices between them and two blocks are on the distance $5$.
	Since $(v_6 v_7,t)$ is the last (rightmost) time-edge in $\mathcal{Y}$
	it is on the distance $6$ from the beginning of the block
	and therefore on the distance $11$ from $(u_6 u_7, t)$ .
	
	If there are $k$ segment blocks between $\mathcal{X}$ and $\mathcal{Y}$, there are $8 k + 4 (k+1)$ vertices between the end of $\mathcal{X}$
	and the beginning of the $\mathcal{Y}$.
	Since $(v_6 v_7,t)$ is the last (rightmost) time-edge in $\mathcal{Y}$
	there are $6$ other vertices before it in the block.
	Therefore, there are $12k + 10$ vertices
	between
	the two edges $(u_6 u_7, t)$ and $(v_6 v_7,t)$, 
	so they are on the distance $2(6k +5) +1$.
	
	Similarly holds for the leftmost time-edges $(u_0 u_1, t)$ and $(v_0 v_1,t)$ of two segment blocks $\mathcal{X}$ and $\mathcal{Y}$ in $\mathcal{G_\phi}$.
\end{proof}

\subsubsection*{Vertical Line Gadget}
A vertical line gadget is used to connect a variable gadget to the appropriate clause gadget.
It consists of one edge appearing in $2k$ consecutive time-steps, where $k$ is a positive integer. The 
value of $k$ for each clause gadget will be specified later; it depends on the embedding $R_{\phi}$ of the input formula $\phi$. 
More precisely, let $\mathcal{X}$ be a segment block for the literal $x_i$ (resp.~$\overline{x_i}$) on vertices $(u_0, u_1, \dots , u_7)$, which starts at time $t$ and finishes at time $t' = t+8$. 
Then the \emph{vertical line gadget} $\mathcal{V}$ of this segment block consists of just the $2k$ appearances 
of the edge $u_6u_7$ from time $t'$ to time $t' + 2k -1$ (resp.~of the edge $u_0u_1$ from time $t-2k +1$ to time $t$).

\begin{lemma}\label{lem:NP-verticalLineGadet}
	Let $\mathcal{V}$ be any vertical line gadget, whose edge appearances are between time $t_0$ and time $t_0 + 2k -1$.
	Then a minimum $2$-TVC of $\mathcal{V}$ in the time windows 
	from $W_{t_0}$ to $W_{t_0 + 2k - 2}$ 
	is of size $k$.
\end{lemma}

\begin{proof}
	To cover an edge in a specific time-step we use just one vertex appearance.
	Since $\mathcal{V}$ consists of $2k$ consecutive appearances of the same edge $u_xu_{x+1}$, 
	containing in the 2-TVC an appearance of $u_x$ or $u_{x+1}$ at a time $t$ (where $t_0+1 \leq t\leq t_0+2k-2$) 
	temporally covers the edge $u_xu_{x+1}$ in two time windows, namely $W_{t-1}$ and $W_{t}$. 
	Note that there are $2k -1$ time windows between $W_{t_0}$ to $W_{t_{0} + 2k - 2}$. 

	 
	Suppose that there is an optimum $2$-TVC of $\mathcal{V}$ of size at most $k-1$.
	This $2$-TVC then can cover the edge $u_xu_{x+1}$ in at most $2(k-1)$ time windows,
	which leaves at least one time window uncovered.
	Therefore the size of an optimum $2$-TVC is at least $k$.
	
	Let us now build two minimum $2$-TVCs of $\mathcal{V}$, both of size $k$.
	We observe the following two options.
	\begin{itemize}
		\item Suppose we cover $u_xu_{x+1}$ at time $t_0$. Then we can cover it also at times $t_0+2, t_0+4, \dots, t_0+2k-2$, and thus $u_xu_{x+1}$ is covered in all time windows between $W_{t_0}$ to $W_{t_{0} + 2k - 2}$ by using exactly $k$ vertex appearances.
		This $2$-TVC corresponds to the red-colored vertex appearances of the vertical line gadget depicted in the \cref{fig:NP-verticalLineGadgetsCover}.
		\item Suppose we cover $u_xu_{x+1}$ at time $t_0+1$ (instead of time $t_0$).
		Then we can cover it also at times $t_0+3, t_0+5, \dots, t_0+2k-1$, and thus $u_xu_{x+1}$ is covered again in all time windows between $W_{t_0}$ to $W_{t_{0} + 2k - 2}$ by using exactly $k$ vertex appearances.
		This $2$-TVC corresponds to the green-colored vertex appearances of the vertical line gadget depicted in the \cref{fig:NP-verticalLineGadgetsCover}.
	\end{itemize} 
\end{proof}

	\begin{figure}[h]
	\centering
	\includegraphics[width=0.22\linewidth]{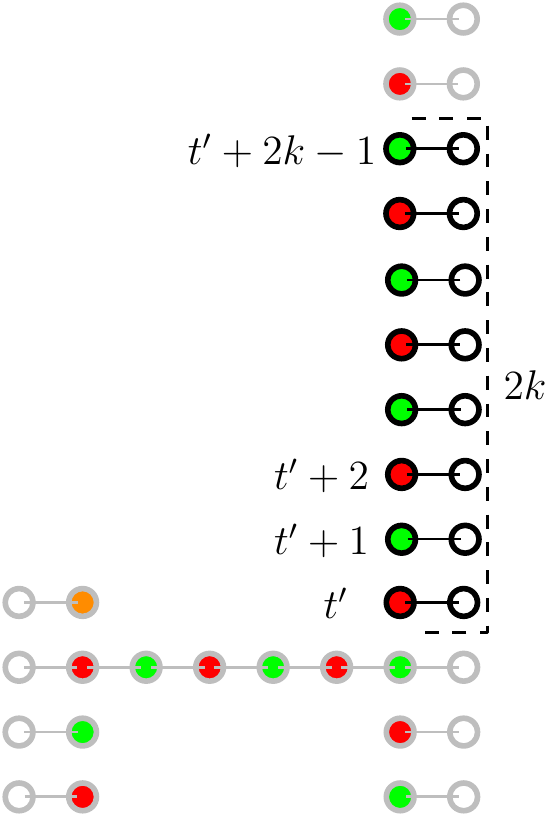}
	\caption{An example of two optimum covers of a vertical line gadget: (i)~with the green-colored, or 
		(ii)~with the red-colored vertex appearances.}
	\label{fig:NP-verticalLineGadgetsCover}
\end{figure}

\subsubsection*{Clause Gadget}
Without loss of the generality, we can assume that in every 
positive clause $(x_i \lor x_j \lor x_k)$ (resp.~negative clause $(\overline{x_i} \lor \overline{x_j} \lor \overline{x_k})$) 
the literals are ordered such that 
$i \leq j \leq k$.
Every clause gadget corresponding to a positive 
(resp.~negative
) clause consists of:
\begin{itemize}
	\item three edges appearing in $4$ consecutive time-steps,
these edges correspond to the rightmost (leftmost, in the case of a negative clause) edge of a segment block of each variable,
	\item two paths $P_1$ and $P_2$, each of odd length, such that 
$P_1$ (resp.~$P_2$) connects the second-to-top (or~second-to-bottom, in the case of a negative clause) newly added edges above (below, in the case of a negative clause) the segment blocks of $x_i$ and $x_j$ (resp.~of $x_j$ and $x_k$). 
We call paths $P_1$ and $P_2$ \emph{clause gadget connectors}.
\end{itemize}
For an example see \cref{fig:NP-clauseCovering}.

\begin{figure}[h]
	\centering
	\includegraphics[width=0.5\linewidth]{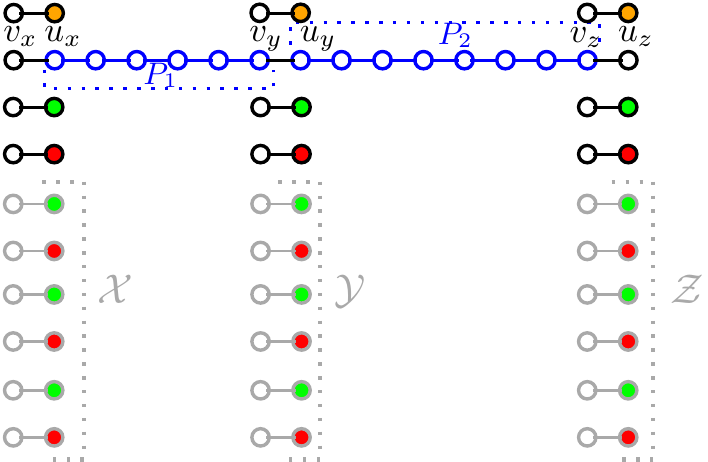}
	\caption{An example of a positive clause gadget, depicted with vertical line gadgets $\mathcal{X, Y, Z}$ connected to it.}
	\label{fig:NP-clauseCovering}
\end{figure}

\begin{lemma}\label{lem:NP-reductionTVCclauseEdge}
Let $C$ be a positive clause $(x_i \lor x_j \lor x_k)$ (resp.~negative clause $(\overline{x_i} \lor \overline{x_j} \lor \overline{x_k})$)
and let $\mathcal{C}$ be its corresponding clause gadget in $\mathcal{G_\phi}$. 
Let the clause gadget connectors $P_1$ and $P_2$ have lengths $p_1$ and $p_2$, respectively. 
If we temporally cover at least one of the three vertical line gadgets connecting to $\mathcal{C}$ with the green (resp.~red) vertex appearances, 
then $\mathcal{C}$ can be covered with exactly $7+\frac{p_1+p_2}{2}$ vertex appearances; otherwise we need at least $8+\frac{p_1+p_2}{2}$ vertex appearances to temporally cover $\mathcal{C}$.
\end{lemma}

\begin{proof}
	Let $\mathcal{X, Y, Z}$ be the vertical line gadgets connecting the segment blocks of variables $x_i, x_j, x_k$ to the clause gadget, respectively 
	and let $C$ be a positive clause.
	Suppose that $\mathcal{C}$ starts at time $t$.
	We denote with $v_x u_x, v_y u_y, v_z u_z$ the underlying edges of
	$\mathcal{X},\mathcal{Y} ,\mathcal{Z}$ respectively.
	The clause gadget consists of edges $v_x u_x, v_y u_y, v_z u_z$ appearing in consecutive time-steps from $t$ to $t+3$,
	together with
	the clause gadget connectors, \ie
	a path $P_1$ of length $p_1$ from $u_x$ to $v_y$ at time $t+2$ and
	a path $P_2$ of length $p_2$ from $u_y$ to $v_z$ at time $t+2$.
	Note $p_1$ and $p_2$ are odd (see \cref{lem:NPoddDistance-RightmostEdges}).
	
	A $2$-TVC of $\mathcal{C}$ 
	always covers all three time-edges at time $t+3$, as the time window $W_{t+3}$ needs to be satisfied.	
	If $\mathcal{X}$ is covered with the red vertex appearances, then 
	we need to cover its underlying edge at time $t$, to satisfy the time window $W_{t-1}$.
	Besides that, we also need one extra vertex appearance for the time window $W_{t+1}$. 
	Without loss of generality, we can cover the edge at time $t+2$.
	If $\mathcal{X}$ is covered with the green vertex appearances,
	then the edge $v_xu_x$ is already covered in the time window $W_{t-1}$, so we cover it at time $t+1$, 
	which satisfies the remaining time windows $W_t$ and $W_{t+1}$.
	Similarly it holds for $\mathcal{Y} ,\mathcal{Z}$.
	
	We still need to cover the clause gadget connectors $P_1$ and $P_2$.
	Since $P_1$ and $P_2$ are paths of odd length appearing at only one time-step,
	their optimum $2$-TVC is of size $(p_1 + 1) /2$ and $(p_2 + 1) /2$, respectively.
	Therefore, if one endpoint of $P_1$ or $P_2$ is in the $2$-TVC, then to cover them optimally the other endpoint cannot be in the cover.
	
	If $\mathcal{Y}$ is covered with the red $2$-TVC, then one of the endpoints of the time-edge $(v_y u_y, t + 2)$
	has to be in the $2$-TVC,
	in the other case, the time-edge is already covered by a vertex appearance in the previous time-step.
	If $\mathcal{X}$ (resp.~  $\mathcal{Z}$) is covered with the red $2$-TVC, the first (resp.~ last) time-edge of $P_1$ (resp.~$P_2$) has to be covered,
	thus vertex appearance $(u_x, t + 2)$ (resp.~$(v_z,t+2)$) is in the $2$-TVC.
	Therefore if all $\mathcal{X,Y,Z}$ are covered with the red $2$-TVC 
	we use $(p_1 + p_2) /2 +2$ vertex appearances to cover $P_1$ and $P_2$.

	Altogether we need
	\begin{itemize}
		\item one vertex appearance for each edge at time $t+3$,
		\item one vertex appearance for each edge at time-steps $t$ and $t+1$ and
		\item either $(p_1 + p_2) /2 +1$ or $(p_1 + p_2) /2 +2$ vertex appearances to cover the clause gadget connectors.
	\end{itemize}
	We have $8 + \frac{p_1 + p_2}{2}$ vertex appearances to cover $\mathcal{C}$ if all three 
	$\mathcal{X,Y,Z}$ are covered with the ``orange and red'' $2$-TVC 
	and $7 + \frac{p_1 + p_2}{2}$ in any other case.
	
	Similarly, it holds when $C$ is a negative clause.
\end{proof}
\fi

\begin{figure}[h!]
	\centering
	\includegraphics[width=0.8\linewidth]{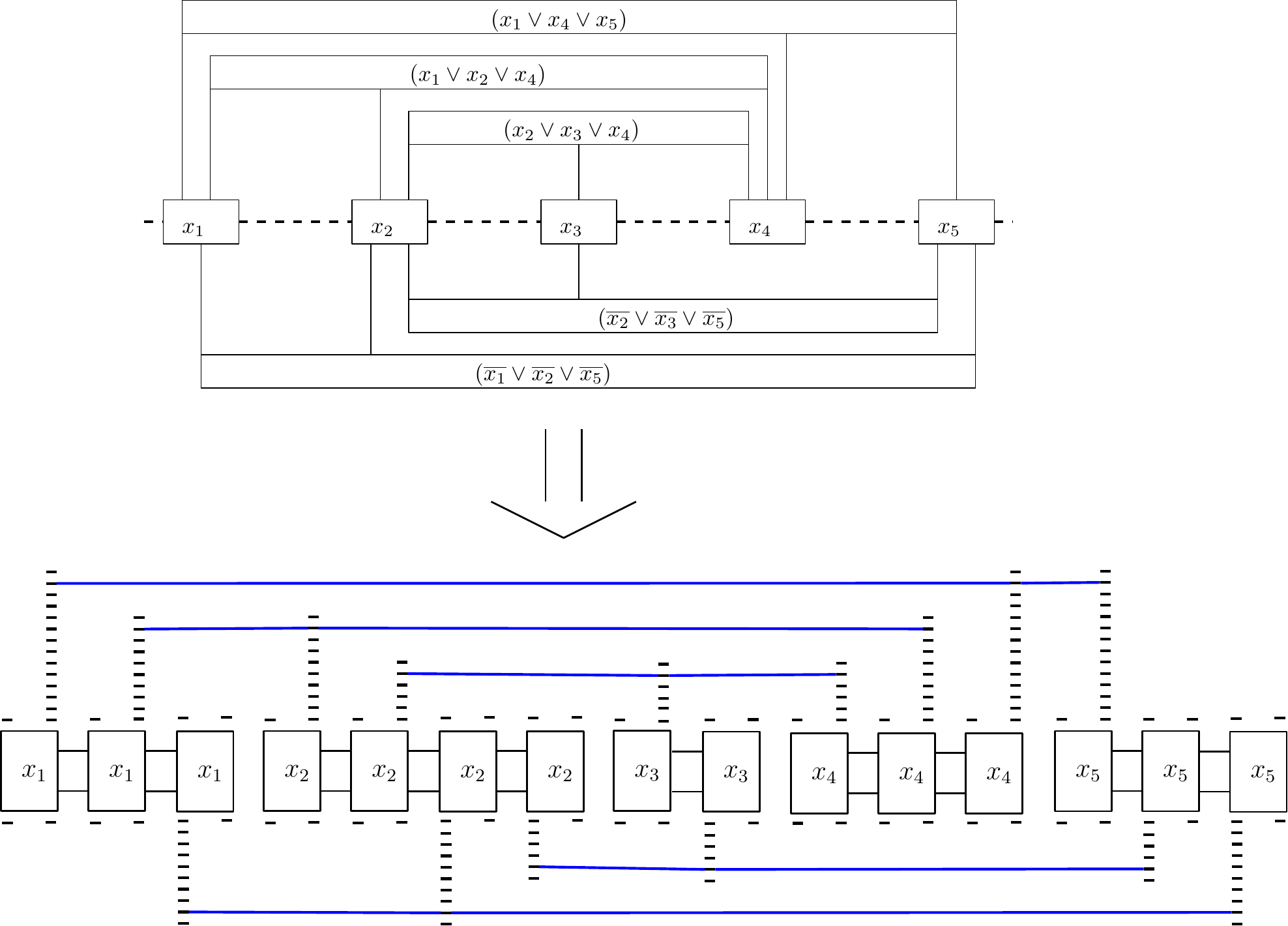}
	\caption{An example of the construction of a temporal graph from a planar rectilinear embedding of monotone 3SAT.}
	\label{fig:NP-rectilinearConstruction}
\end{figure}

\iflong
\subsubsection*{Creating a Temporal Path $\mathcal{G_\phi}$}
Before we start with the detailed construction we need to fix the following notation.
\paragraph{Notation.}
From the planar, rectilinear embedding $R_\phi$ of $\phi$, we can easily fix the order of clauses.
We fix the clauses by first fixing all the positive clauses and then all negative ones using the order they appear in the $R_\phi$.
If $C_i, C_j$ are two positive (resp.~negative) clauses and $i < j$, then the clause $C_i$ is above (resp.~below) the clause $C_j$ in the $R_\phi$.
For example in \cref{fig:rectilinearMax2SAT} we have
$C_1 = (x_1 \lor x_4 \lor x_5), C_2 = (x_1 \lor x_2 \lor x_4), C_3 = (x_2 \lor x_3 \lor x_4), C_4 = (\overline{x_1} \lor \overline{x_2} \lor \overline{x_5}), C_5 = (\overline{x_2} \lor \overline{x_3} \lor \overline{x_5})$.

Recall that $d_i$ denotes the number of appearances of the variable $x_i$ as a literal (\ie either $x_i$ or $\overline{x_i}$) in $\phi$. 
For a clause $C_a=(x_i \lor x_j \lor x_k)$ or $C_a = (\overline{x_i} \lor \overline{x_j} \lor \overline{x_k})$ 
we denote with $g^a_i, g^a_j$ and $g^a_k$
the appearances of the literals $x_i,x_j$ and $x_k$ (resp.~$\overline{x_i},\overline{x_j}$ and $\overline{x_k}$) in $C_a$,
\ie $x_i$ (resp.~$\overline{x_i}$) appears in $C_a$ for the $g^a_j$-th time, 
$x_j$ (resp.~$\overline{x_j}$) appears $g^a_j$-th time and
$x_k$ (resp.~$\overline{x_k}$) appears $g^a_k$-th time.
With $m_1$ we denote the number of positive and with $m_2$ the number of negative clauses in $\phi$.
For every clause $C_a$, we define also the level $\ell_a$ in the following way.
If $C_a$ is a positive clause then $\ell_a$ equals the number of clauses between it and the variable-axis increased by $1$, 
\ie the positive clause $C_1$, which is furthest away from the variable-axis, is on the level $m_1$, the clause $C_{m_1}$, which is the closest to the variable axis, is on level $1$.
Similarly, for negative clauses, the level is negative. For instance, the clause $C_{m_1 + 1}$ is on the level $- m_2$ and the clause $C_{m}$ is on the level $-1$.

Now we are ready to combine all the above gadgets and constructions to describe the construction of the temporal path $\mathcal{G_\phi}$.

\paragraph{Detailed construction.}
For each variable $x_i$ in $\phi$
we construct its corresponding variable gadget in $\mathcal{G_\phi}$. 
All variable gadgets lie on the same horizontal line.
We create the clause gadget of a positive $C_a=(x_i \lor x_j \lor x_k)$ or negative clause $C_a=(\overline{x_i} \lor \overline{x_j} \lor \overline{x_k})$ 
such that we connect the
$g^a_i$-th copy of the segment block of $x_i$ to it, together with the $g^a_j$-th copy of the segment block of $x_j$ and
$g^a_k$-th copy of the segment block of $x_k$.
We do this using vertical line gadgets,
that connect the corresponding segment blocks to the level
$\ell_a$ where the clause gadget is starting.

Let us now define our temporal graph $\mathcal{G_\phi}=(G,\lambda)$.
\begin{itemize}
	\item The underlying graph $G$ is a path on $12 \sum_{i=1}^{n}d_i - 4$ vertices.
	\begin{itemize}
		\item
		A variable gadget corresponding to the variable $x_i$ consists of $d_i$ copies of the segment block, together with $d_i -1$ horizontal bridges.
		We need $d_i  8 + (d_i-1)  4 = 12 d_i - 4$ vertices for one variable and $\sum_{i=1}^{n} (12d_i - 4)$ for all of them. 
		\item Between variable gadgets of two consecutive variables $x_i, x_{i+1}$ there are $4$ vertices.
		Therefore we use extra $4(n-1)$ vertices, for this construction.
	\end{itemize}
	\item The lifetime $T$ of $\mathcal{G_\phi}$ is $4(m+4)$.
	\begin{itemize}
		\item There are $m$ clauses and each of them lies on its own level.
		Therefore we need $4m$ time-steps for all clause gadgets.
		\item 
		There is a level $0$ with all the variable gadgets,
		of height $9$.
		We add one extra time-step before we define the start of levels $1$ and $-1$.
		This ensures that the vertical line gadget, connecting the variable gadget to the corresponding clause gadget, is of even height.
		\item The first time-edge appears at time $5$ instead of $1$,
		which adds $4$ extra time-steps to the lifetime of $\mathcal{G_\phi}$ and does not 
		interfere with our construction.
		Similarly, after the appearance of the last time-edge, there is one extra time-step, where no edge appears.
		We add these time-steps to ensure that the dummy time-edges must be covered in the time-step they appear
		and that the lifetime of the graph $\mathcal{G_\phi}$ is divisible by $4$.
		\item All together we have $4 m + 11 + 4 + 1 $ time-steps.
	\end{itemize}
\end{itemize}

Therefore $\mathcal{G}_\Phi$ is a temporal graph of lifetime $T = 4(m+4)$ with the underlying path on $12 \sum_{i=1}^{n}d_i -4$ vertices.
The appearances of each edge arise from the structure of the planar monotone rectilinear 3SAT.

\subsubsection*{\boldmath Size of the Optimum $2$-TVC of $\mathcal{G_\phi}$}

Using the notation introduced above we determine the size of the optimum $2$-TVC of $\mathcal{G_\phi}$.
We do this in two steps, first determining the size of the optimum $2$-TVC for each variable gadget
and then determining the optimum $2$-TVC of all vertical line gadgets together with the clause gadgets.
\begin{enumerate}
	\item Optimum $2$-TVC covering all variable gadgets is of size $19 \sum _{i=1}^n d_i -4n$.
	\begin{itemize}
		\item Since a variable $x_i$ appears in $d_i$ clauses as a positive or negative literal,
		we construct $d_i$ connected copies of the segment block.
		By \cref{cor:NP-SizeOfTVCinHorizontalDirection} this construction has exactly two optimum $2$-TVCs,
		each of size $19 d_i - 4$.
		\item Using one of the four dummy time-edges of a segment block, denote it $(e,t)$, and we connect the corresponding variable gadget to the clause gadget.
		Therefore we need to cover just $3$ dummy time-edges.
		Since there is one extra time-step before the first clause gadget level, the edge $e$ appears also at time $t+1$ in the case of a positive clause and $t-1$ in the case of a negative one.
		To optimally cover it we then use one vertex appearance.		
		\item There are $n$ variables, so all together we need $\sum _{i=1}^n (19d_i -4)$ vertex appearances in $2$-TVC to optimally cover all of them.
	\end{itemize}
	\item Optimum $2$-TVC covering all vertical line gadgets and clause gadgets is of size	  
\begin{equation}\label{eq:NP-TVCsizeOfAllClauses}
\begin{aligned} 
&
\sum_{\substack{C_a = (x_i \lor x_j \lor x_k) \text{ or}\\ 
		C_a=(\overline{x_i} \lor \overline{x_j} \lor \overline{x_k})}} 
\Bigg( 6|\ell_a| + 6 \sum_{b=1}^{k-i-1} d_{i + b} +  \\
& 6(d_i - g^a_i + g^a_k) + k - i \Bigg) - 2m.
\end{aligned}
\end{equation}

	Let $C_a = (x_i \lor x_j \lor x_k)$ (resp~$C_a = (\overline{x_i} \lor \overline{x_j}\lor \overline{x_k})$).
	\begin{itemize}
		\item Vertical line gadgets connect variable gadgets of $x_i, x_j$ and $x_k$ with the clause gadget of $C_a$.
		Since the clause gadget is on level $\ell_a$
		the corresponding vertical line gadgets are of height $4(|\ell_a| -1)$.
		By \cref{lem:NP-verticalLineGadet} and the fact that each clause requires exactly three vertical line gadgets, we need $6(|\ell_a| -1)$ vertex appearances to cover them optimally.	
		\item 
		Using \cref{lem:NP-reductionTVCclauseEdge} we know that an optimum $2$-TVC covering the clause gadget corresponding to $C_a$ requires $7 + \frac{p_1 + p_2}{2}$ vertex appearances where
		$p_1$ and $p_2$ are the lengths of the clause gadget connectors.
		Let us determine the exact value of $p_1 + p_2$.
		\begin{itemize}
			\item There are $k-i-1$ variable gadgets between $x_i$ and $x_k$, namely $x_{i+1}, x_{i+2}, \dots , x_{k-1}$.
			Variable gadget corresponding to any variable $x_b$ is of length $12 d_b -4$.
			There are $4$ vertices between two consecutive variable gadgets.
			Therefore there are $\sum_{d=1}^{k-i-1} (d_{i + d} 12 - 4) + (j-i)4 = 12 \sum_{d=1}^{k-i-1} d_{i + d} -4$
			vertices between the last vertex of $x_i$ variable gadget and the first vertex of $x_k$ variable gadget.
			\item From the end of $g^a_i$-th segment block of $x_i$ to the 
			end of the variable gadget
			there are $d_i - g^a_i $ copies of the segment block together with the $d_i - g^a_i$ horizontal bridges connecting them.
			Similarly, the same holds for the case of $x_k$, where there are
			$g^a_k - 1$ segment blocks and horizontal bridges.
			Altogether we have 
			$(d_i - g^a_i + g^a_k - 1)8 + (d_i - g^a_i + g^a_k - 1)4  = 12 (d_i - g^a_i + g^a_k - 1)$ vertices.
			\item Between the end of $g^a_i$-th segment block of $x_i$ and the beginning of the $g^a_k$-th segment block of $x_k$ we have
			\begin{align*}
			p_1 + p_2 = & 12 \sum_{d=1}^{k-i-1} d_{i + d} +4(k-i+1)+ \\
			& 12 (d_i - g^a_i + g^a_k - 1)
			\end{align*}
			 
			vertices.
		\end{itemize}	
		\end{itemize}
		All together the optimum $2$-TVC of the clause $C_a$ with the corresponding vertical line gadgets 
		is of size
		\begin{equation*}
		6|\ell_a| + 6 \sum_{b=1}^{k-i-1} d_{i + b} + 
		6(d_i - g^a_i + g^a_k) + k - i - 1.
		\end{equation*}
		
	Extending the above equation to all clauses and variables we get that the optimum $2$-TVC of $\mathcal{G_\phi}$ is of size
	\begin{equation} \label{eq:NP-TVCsizeSatisfiable}
	\begin{aligned} 
s =
& 19 \sum _{i=1}^n d_i +
\sum_{\substack{C_a = (x_i \lor x_j \lor x_k) \text{ or}\\ 
		C_a=(\overline{x_i} \lor \overline{x_j} \lor \overline{x_k})}} 
\Bigg( 6|\ell_a| + 6 \sum_{b=1}^{k-i-1} d_{i + b} +  \\
& 6(d_i - g^a_i + g^a_k) + k - i \Bigg) - m - 4n.
	\end{aligned}
	\end{equation}
\end{enumerate}

We are now ready to prove our main technical result of this section.
\begin{lemma}\label{lem:two-directions}
There exists a truth assignment $\tau$ of the variables in $X$ which satisfies $\phi(X)$ if and only if there exists 
a feasible solution to $2$-TVC on $\mathcal{G_\phi}$ which is of the size $s$ as in the \cref{eq:NP-TVCsizeSatisfiable}.
\end{lemma}

\begin{proof}
Each segment block corresponding to a variable $x$ has exactly two optimum $2$-TVCs (the ``orange and green'' and ``orange and red'' one, depicted in \cref{fig:NP-varibleGadgetsCover}).
We set the ``orange and green'' solution to the \textsc{True} value of $x$
and the ``orange and red'' one to \textsc{False}.

($\Rightarrow$):
From $\tau$ we start building a $2$-TVC of $\mathcal{G_\phi}$ by 
first covering all variable gadgets 
with a $2$-TVC 
of ``orange and green'' (resp.~``orange and red'') vertex appearances 
if the corresponding variable is \textsc{True} (resp.~\textsc{False}).
Using \cref{lem:SizeOfTVCinHorizontalDirection}
we know that the optimum $2$-TVC of every variable gadget
uses either all ``orange and green'' or all ``orange and red'' vertex appearances 
for every segment block and all horizontal bridges connecting them.
Next, we extend the $2$-TVC from variable gadgets to vertical line gadgets,
\ie if the variable gadget is covered with the
``orange and green'' (resp.~``orange and red'')
vertex appearances then the 
vertical line gadget connecting it to the clause gadget uses the green (resp.~red)
vertex appearances.
In the end, we have to cover also clause gadgets.
Using \cref{lem:NP-reductionTVCclauseEdge} we know that we can cover
each clause gadget optimally 
if and only if at least one vertical line gadget connecting to it is covered with the
green vertex appearances, in the case of a positive clause and
red vertex appearances in the case of the negative one.
More precisely, a clause gadget $\mathcal{C}$ corresponding to the clause $C$
is covered with the minimum number of vertex appearances whenever $C$ is satisfied.
Since $\tau$ is a truth assignment satisfying $\phi$, all clauses are satisfied and hence covered optimally.
Therefore the size of the $2$-TVC of $\mathcal{G_\phi}$ 
is the same as in the \cref{eq:NP-TVCsizeSatisfiable}.

($\Leftarrow$):
Suppose now that $\mathcal{C}$ is a minimum $2$-TVC of the temporal graph $\mathcal{G_\phi}$ of size $s$. 
Using \cref{lem:SizeOfTVCinHorizontalDirection,cor:NP-SizeOfTVCinHorizontalDirection}
we know that the optimum $2$-TVC of each variable gadget consists of either all ``orange and green'' or all ``orange and red'' vertex appearances.
If this does not hold then $\mathcal{C}$ is not of minimum size on variable gadgets. 
Similarly, the \cref{lem:NP-reductionTVCclauseEdge} 
assures that the number of vertex appearances used in the $2$-TVC of any clause gadget
is minimum if and only if there is at least one
vertical line gadget covered with
green vertex appearances in the case of a positive clause
and red vertex appearances in the case of a negative clause.
If this does not hold then $\mathcal{C}$ is not of minimum size on clause gadgets.
Lastly, if a vertical line gadget $\mathcal{V}$ 
connects a variable gadget covered with ``orange and green'' (resp.~``orange and red'')  vertex appearances 
to a clause gadget using ``orange and red'' (resp.~``orange and green'') vertex appearances,
then the $2$-TVC of $\mathcal{V}$ is not minimum (see \cref{lem:NP-verticalLineGadet}) and
$\mathcal{C}$ cannot be of size $s$.

To obtain a satisfying truth assignment $\tau$ of $\phi$ we
check for each variable $x_i$ which vertex appearances are used in the $2$-TVC of its corresponding variable gadget.
If $\mathcal{C}$ covers the variable gadget of $x_i$ with all ``orange and green'' vertex appearances we set $x_i$ to  \textsc{True} else we set it to \textsc{False}.
\end{proof}

Our main result of this section follows now directly from Lemma~\ref{lem:two-directions} and the fact that the 
planar monotone rectilinear 3SAT is NP-hard \cite{Berg2010Optimal}.
\fi
As, for every $\Delta\geq 2$, there is a known polynomial-time reduction from $\Delta$-TVC to $(\Delta+1)$-TVC~\cite{AkridaMSZ20}, we obtain the following.

\begin{theorem}\label{thm:2tvc-paths-NP}
	For every $\Delta\geq 2$, $\Delta$-TVC on instances on an underlying path is NP-hard.
\end{theorem}

With a slight modification to the $\mathcal{G_\phi}$ we can create the temporal cycle from $R_\phi$ and therefore the following holds.

\begin{corollary}\label{cor:2tvc-cycles-NP}
	For every $\Delta\geq 2$, $\Delta$-TVC on instances on an underlying cycle is NP-hard.
\end{corollary}

\begin{proof}
We follow the same procedure as above where we add one extra vertex $w$, to the underlying graph $P$ of $\mathcal{G_\phi}$.
We add also two time-edges connecting the first and the last vertex of the temporal path graph $\mathcal{G_\phi}$ at time $1$.
This increases the size of the $2$-TVC by $1$ (as we need to include the vertex appearance $(w,1)$) and it transforms the underlying path $P$ into a cycle.
\end{proof}

\subsection{Algorithmic Results}
\label{subsec:algres}
To complement the hardness presented in \Cref{subsec:2tvc-paths-NP}, we present two polynomial-time algorithms.
Firstly, a dynamic program for solving \textsc{TVC} on instances with underlying paths and cycles, which shows that the hardness is inherently linked to the sliding time windows.
Secondly, we give a PTAS for \dTVC\ on instances with underlying paths.
This approximation scheme can be obtained using a powerful general purpose result commonly used for approximating geometric problems~\cite{mustafa2010hittingset}.

\begin{theorem}
TVC can be solved in polynomial time on instances with a path/cycle as their underlying graph.
\end{theorem}
\iflong
\begin{proof}
	Let us start the proof with the case when the underlying graph $G = (v_0, v_1, v_2, \dots, v_n)$ is a path on $n+1$ vertices. 
	Our algorithm mimics the greedy algorithm for computing a minimum vertex cover in a (static) path graph, 
	which just sweeps the path from left to right and takes every second vertex. 
	For every $i=1,2,\ldots, n$ we denote by $e_i$ the edge $v_{i-1}v_i$.
	
	To compute the optimum temporal vertex cover $\mathcal{C}$ of $\mathcal{G}$ we need to determine for each edge which vertex appearance covers it. 
	We first consider the edge $v_0 v_1$. Without loss of generality, we can cover this edge in $\mathcal{C}$ by a vertex appearance 
	in the set $\{(v_1,t) : v_0 v_1 \in E_t\}$, as no appearance of vertex $v_0$ can cover any other edge apart from $v_0 v_1$. 
	We check whether the set $\{t : v_0 v_1 \in E_t\} \cap \{t' : v_1 v_2 \in E_{t'}\}$ is empty or not; 
	that is, we check whether $v_0 v_1$ appears together with $v_1 v_2$ in any time-step. 
	If this set is empty, then we add to $\mathcal{C}$ an arbitrary vertex appearance $(v_1,t)$, where $(e_1,t) \in E_t$. 
	This choice is clearly optimum, as we need to cover $v_0 v_1$ with at least one vertex appearance of $v_1$, 
	while any of the vertex appearances of $v_1$ can not cover any other edge of the graph. 
	Otherwise, let $t\in \{t : v_0 v_1 \in E_t\} \cap \{t' : v_1 v_2 \in E_{t'}\}$ be an arbitrary time-step in which both $v_0 v_1$ and $v_1 v_2$ appear. 
	Then we add $(v_1,t)$ to $\mathcal{C}$. 
	This choice is again optimum, as in this case the vertex appearance $(v_1,t)$ covers both edges $v_0 v_1$ and $v_1 v_2$. 
	After adding one vertex appearance to $\mathcal{C}$, we ignore the edges that have been covered so far by $\mathcal{C}$, and we proceed recursively with 
	the remaining part of $\mathcal{G}$, until all edges are covered.
	Each iteration of the above algorithm can be performed in $\mathcal{O}(T)$ time, and we have at most $\mathcal{O}(n)$ iterations. 
	Therefore the running time of the whole algorithm is $\mathcal{O}(Tn)$.
	
	We can similarly deal with the case where the underlying graph $G$ is a cycle on the vertices $v_0, v_1, v_2, \dots, v_n$. 
	Here we just need to observe that, without loss of generality, the optimum solution contains either a vertex appearance of $v_0$ 
	or a vertex appearance of $v_1$ (but noth both). Therefore, to solve TVC on $\mathcal{G}$, we proceed as follows. 
	First we create an auxiliary cycle $G'$ on the vertices $v_0', v_0, v_1, v_2, \ldots, v_n$ (in this order), where $v_0'$ is a new dummy vertex. 
	Then we create the temporal graph $\mathcal{G}'=(G',\lambda')$, where $\lambda'(v_i v_{i+1}) = \lambda(v_i v_{i+1})$ for every $i=1,2,\ldots,n-1$, 
	$\lambda'(v_n v_0') = \lambda(v_n v_0)$, and $\lambda'(v_0' v_0) = \lambda'(v_0 v_1) = \lambda(v_0 v_1)$. 
	Next we compute the temporal subgraphs $\mathcal{G}_1'$ and $\mathcal{G}_2'$ of $\mathcal{G}'$, where the underlying graph of 
	$\mathcal{G}_1'$ (resp.~of $\mathcal{G}_2'$) is the path $(v_0, v_1, \dots, v_n, v_0')$ (resp.~the path $(v_1, v_2, \dots, v_n, v_0', v_0)$). 
	Every edge of $\mathcal{G}_1'$ and of $\mathcal{G}_2'$ has the same time labels as the corresponding edge in $\mathcal{G}$. 
	Now we optimally solve TVC on $\mathcal{G}_1'$ and on $\mathcal{G}_2'$, and we return the smaller of the two solutions, 
	as an optimum solution of TVC on $\mathcal{G}$ (by replacing any vertex appearance $(v_0',t)$ by $(v_0,t)$). 
	The running time is again~$\mathcal{O}(Tn)$.
\end{proof}
\fi

Next, we turn to approximating \dTVC\ on underlying paths.
The \textsc{Geometric Hitting Set} problem takes as an instance a set of geometric objects, e.g.\ shapes in the plane, and a set of points.
The task is to find the smallest set of points, that hit all of the objects.
In the paper by Mustafa and Ray \cite{mustafa2010hittingset} the authors present a $\PTAS$
for the problem, when the geometrical objects are $r$-admissible set regions.
We transform an arbitrary temporal path to the setting of the geometric hitting set.
As a result, we obtain a $\PTAS$ for the \dTVC\ problem.

\iflong
\begin{definition}
	Let $R = (P,D)$ be a range space that consists of points $P$ in $\mathbb{R}^2$ and a set $D$ of regions containing points of $P$.
	The \emph{minimum geometric hitting set} problem asks for the smallest set $S \subseteq P$, such that each region in $D$ contains at least one point from $S$.
\end{definition}

\begin{definition}
	Let $D$ be the set of regions in the plane, where any region $s \in D$ is bounded by a closed Jordan curve \footnote{A Jordan curve is a non-self-intersecting continuous loop.}.
	We say, that $D$ is an \emph{$r$-admissible} set of regions
	if Jordan curves bounding any two regions $s_1, s_2 \in D$ cross $l \leq r$ times
	and both $s_1 \setminus s_2$ and $s_2 \setminus s_1$ are connected regions\footnote{The space $X$ is said to be path-connected if there is a path joining any two points in $X$. Every path-connected space is also connected.}.
\end{definition} 

\begin{theorem}[\cite{mustafa2010hittingset}] \label{thm:PTAS}
	Let $R = (P,D)$ be a range space consisting of a set $P$ of $n$ points in $\mathbb{R}^2$ and a set $D$ of $m$ $r$-admissible regions.
	There exists a $(1 + \varepsilon)$-approximation algorithm for the minimum geometric hitting set problem that is performed in  $\mathcal{O}((mn^{\mathcal{O}(\varepsilon^{-2})})$ time.
\end{theorem}
\fi

\begin{theorem}
	For every $\varepsilon>0$, there exists an $(1 + \varepsilon)$-approximation algorithm for \dTVC\ on instances with a path/cycle as their underlying graph, which runs in time
	\begin{align*}
	& O\left(n (T-\Delta + 1) \cdot (T(n+1))^{\mathcal{O}(\varepsilon^{-2})}\right) = \\
	& O\left((T(n+1))^{\mathcal{O}(\varepsilon^{-2})}\right), 
	\end{align*}
	\ie the problem admits a PTAS.
\end{theorem}
\begin{proof}

Let $\mathcal{G}=(G,\lambda)$ be a temporal path, on vertices $\{v_1, v_2, \dotsc, v_n\}$, with lifetime $T$.
We first have to create the range space $R = (P, D)$, where $P \subseteq V \times [T]$ is a set of vertex appearances and $D$ is a set of $2$-admissible regions.

Set $P$ of time vertices consist of vertex appearances $(v_i, t)$ for which $t \in \lambda(e_i) \cup \lambda(e_{i+1})$.
Intuitively, if edges incident to $v$ do not appear at time $t$, then $(v,t)$ is not in $P$.
Set $D$ of $2$-admissible regions consists of rectangles of $2$ different sizes.
For every edge $e_i$ that appears in the time window $W_t$ we create one rectangle $R_i^t$, that includes all vertex appearances incident to $e_i$ in $W_t$.
For example, if edge $e_i$ appears in the time window $W_t$ at times $t_1$ and $t_2$, then the corresponding rectangle $R_i^t$ contains vertex appearances $(v_{i-1},t_1), (v_{i}, t_1), (v_{i-1},t_2)$ and $(v_{i}, t_2)$.

It is not hard to see that $|P| \leq |V| \cdot T = (n+1)T$ and $|D| \leq |E| (T - \Delta + 1) = n (T-\Delta + 1)$.

\subsubsection*{Formal Construction.}
Since $D$ will be defined to be a set of $2$-admissible regions, the boundary of any two rectangles we construct should intersect at most $2$ times.
For this purpose, we use rectangles, of two different sizes.
Let us denote with $A$ the rectangle of size $a_1 \times a_2$ and with $B$ be the rectangle of size $b_1 \times b_2$, where $a_1 > b_1 > b_2 > a_2$.

As we said, for every edge $e_i$ that appears in the time window $W_t$ we construct exactly one rectangle.
These are the rules we use to correctly construct them.
\begin{enumerate}
	\item \label{regions:fixedWindow}
	For a fixed time window $W_t$	
	we construct the rectangles in such a way, that they intersect only in the case when their corresponding edges $e_i, e_j$ share the same endpoint in the underlying graph $G$.
	Since $G$ is a path, the intersection happens only in the case when $j = i+1$.
	We can observe also, that there are no three (or more) edges sharing the same endpoint and therefore no three rectangles intersect.
	
	We require also, that the rectangles corresponding to a pair of adjacent edges are not the same, \ie they alternate between form $A$ and $B$.
	\iflong For an example see \cref{fig:RegionfixedWindow}.
	
	\begin{figure}[h]
			\includegraphics[width=\linewidth]{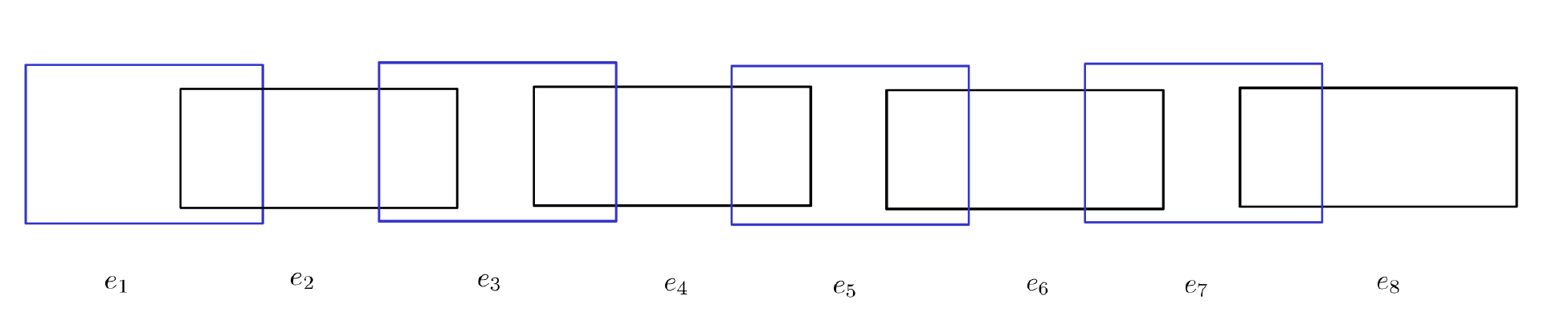}
			\caption{Regions in a fixed time window.}
			\label{fig:RegionfixedWindow}
	\end{figure}	
	\fi
	
	\item \label{regions:fixedEdge}
	For any edge $e_i$, rectangles corresponding to two consecutive time windows $W_t, W_{t+1}$ are not the same, \ie they alternate between the form $A$ and $B$.
	For an example see \cref{fig:RegionfixedEdge}.
	
	\begin{figure}[h]
		\centering
		\begin{subfigure}[t]{0.4\textwidth}
			\centering
			\includegraphics[width=0.45\linewidth]{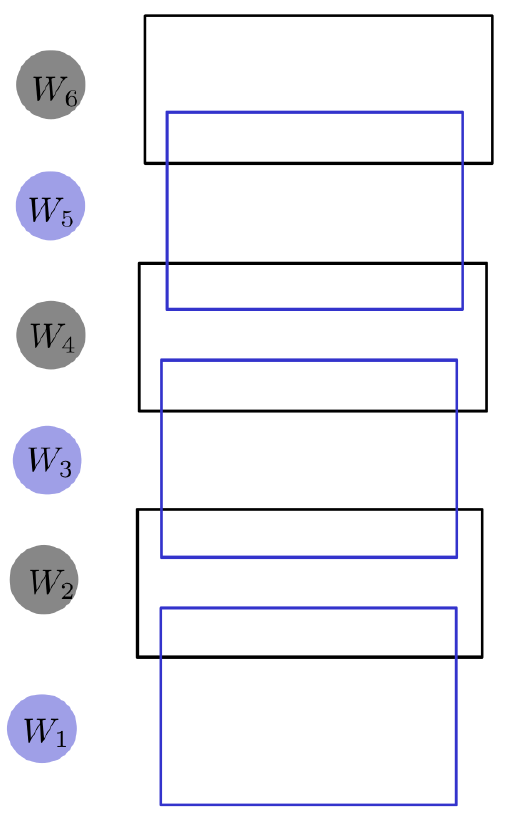}
			\caption{Alternating regions for an edge.}
			\label{fig:RegionfixedEdge} 
		\end{subfigure} \quad
		\begin{subfigure}[t]{0.4\textwidth}
			\centering
			\includegraphics[width=0.7\linewidth]{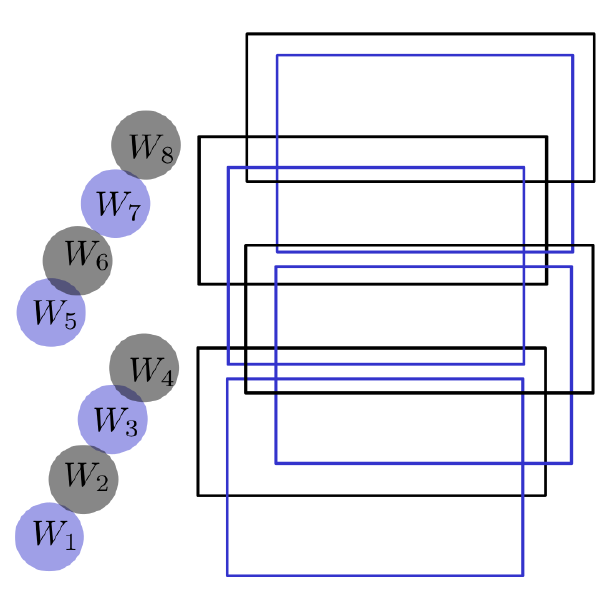}
			\caption{Regions for an edge with shift, when $\Delta = 4$. Each region $W_i$ starts at same horizontal position as $W_{i - 4}$.}
			\label{fig:RegionfixedEdge-delta} 
		\end{subfigure}
		\caption{Creating regions for one edge.}
	\end{figure}

	When the time windows are of size $\Delta$, 
	there are at most $\Delta$ rectangles intersecting at every time-step $t$.
	This holds, because if the edge $e_i$ appears at time $t$ it is a part of the time windows $W_{t_\Delta +1}, W_{t_\Delta +2}, \dots , W_{t}$.
	
	Since the constructed rectangles are of two sizes, if $\Delta \geq 3$ we create intersections with an infinite number of points between the boundary of some rectangles,
	if we just ``stack'' the rectangles upon each other.
	Therefore, we need to shift (in the horizontal direction) rectangles of the same form in one $\Delta$ time window.
	Since the time window $W_t$ never intersects with $W_{t + \Delta}$, we can shift the time windows $W_{t+1}, \dots , W_{t+ \Delta - 1}$ and fix the $W_{t+ \Delta}$ at the same horizontal position as $W_t$ .
	For an example see \cref{fig:RegionfixedEdge-delta}.
\end{enumerate}

Combining both of the above rules we get a grid of rectangles.
Moving along the $x$ axis corresponds to moving through the edges of the path and moving along the $y$ axis corresponds to moving through the time-steps.

\iflong
\begin{figure}[h]
	\centering
		\includegraphics[width=0.8\textwidth]{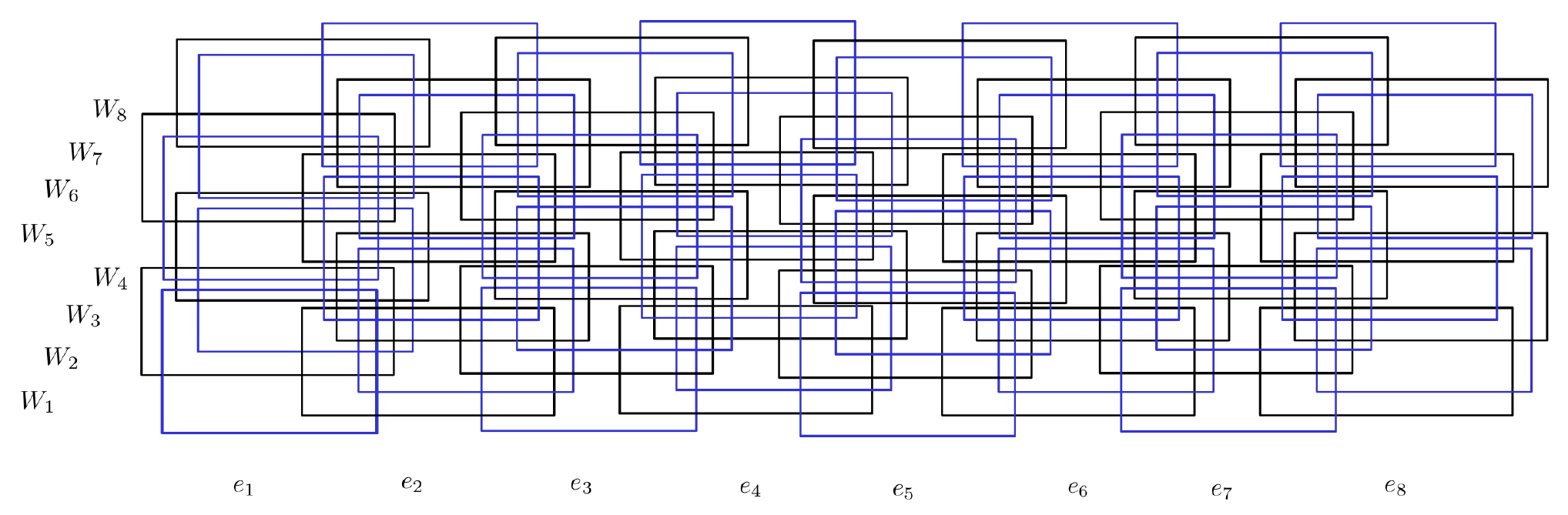}
		\caption{Regions for edges $e_1, \dots, e_8$ in the time windows $W_1, \dots , W_8$, in the case of $4$-TVC.}
		\label{fig:RegionAll}
\end{figure} 
\fi

\iflong 
By \cref{fig:RegionfixedWindow} and \cref{regions:fixedEdge} of the construction above, rectangles alternate between the form $A$ and $B$ in both dimensions. 
See figure \cref{fig:RegionAll} for example. \fi
\ifshort By construction, rectangles alternate between the form $A$ and $B$ in both dimensions. \fi
If an edge $e_i$ does not appear in the time window $W_t$, then we do not construct the corresponding rectangle $R_i^t$.
The absence of a rectangle from a grid does not change the pattern of other rectangles.
To determine what form a rectangle corresponding to edge $e_i$ at time-window $W_t$ is, we use the following condition:

\begin{equation*}
R_i^t = \begin{cases} A & \text{if $i +t  \equiv 0 \pmod 2$,} \\
B & \text{else.} \end{cases}
\end{equation*}

Now we define where the points are placed.
We use the following conditions.
\begin{enumerate}[a.]
	\item If an edge $e_i = v_{i-1} v_i$ appears at time $t$ we add vertex appearances $(v_{i-1},t), (v_i, t)$ to all of the rectangles $R_i^{t'}$, where $t-\Delta +1 \leq t' \leq t$, if they exist.
	
	Equivalently, we add the vertex appearances $(v_{i-1},t), (v_i, t)$ in the intersection of  the rectangles corresponding to the edge $e_i$ in the time windows $W_{t-\Delta +1}, \dots , W_{t}$.
	For an example see \cref{fig:Vertices-edge}.
	
	\item If an edge $e_i$ does not appear at time $t$, then the time vertices
	$(v_{i-1}, t),(v_i,t)$ are not included in the rectangles
	$R_i^{t'}$ ($t-\Delta +1 \leq t' \leq t$), if they exist.
	
	\item If two adjacent edges $e_i, e_{i+1}$ appear at the same time $t$, we add to the intersection of the rectangles $R_i^{t'}, R_{i+1}^{t'}$ the vertex $(v_i,t)$, where $t-\Delta +1 \leq t' \leq t$.
	For an example see \cref{fig:Vertices-twoedge}.
\end{enumerate}

\begin{figure}[htb]
		\centering
		\begin{subfigure}[t]{0.4\textwidth}
			\includegraphics[width=0.8\linewidth]{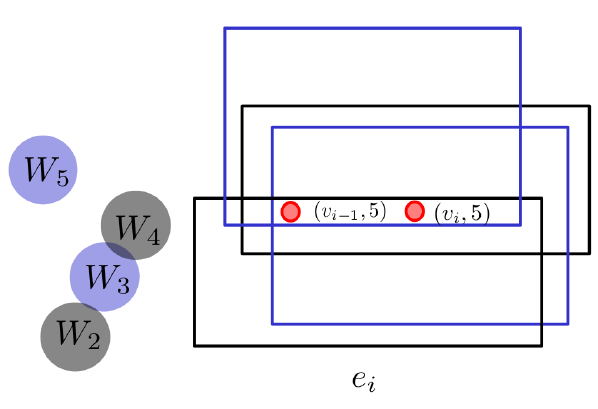}
			\caption{Edge $e_i$ appearing at time $5$ and the corresponding placement of vertex appearances to the rectangles, when $\Delta = 4$.}
			\label{fig:Vertices-edge}
		\end{subfigure} \quad
		\begin{subfigure}[t]{0.4\textwidth}
			\includegraphics[width=\linewidth]{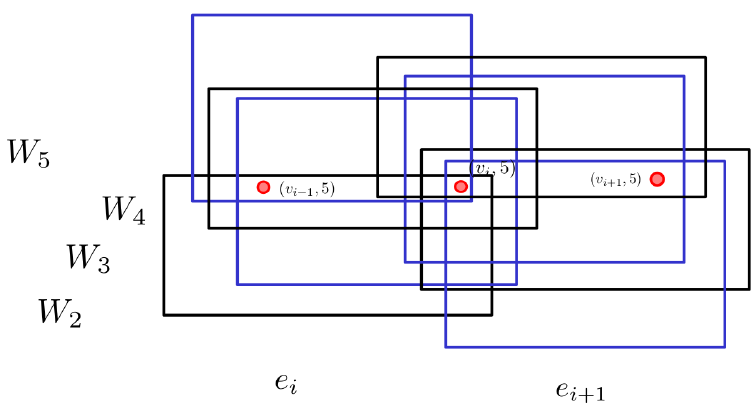}
			\caption{Edges $e_i$ and $e_{i+1}$ appearing at time $5$ and the corresponding vertex appearances placement in the rectangles, when $\Delta= 4$.}
			\label{fig:Vertices-twoedge}
		\end{subfigure}
	\caption{Example of placement of the vertices into the rectangles when $\Delta = 4$.}
\end{figure}
It is straightforward to verify that 
finding the minimum hitting set of the range space is equivalent to finding the minimum \dTVC\ for $\mathcal{G}$.
On the constructed range space we use the local search algorithm from \cite{mustafa2010hittingset}
which proves our result.
	Note, that we assume that if a region does not admit any points, then it is trivially satisfied.
\end{proof}

\section{Algorithms for Bounded Degree Temporal Graphs}\label{bounded-degree-sec}
In this section, we extend our focus from temporal graphs with underlying paths or cycles to instances of \(\Delta\)-TVC with more general degree restrictions.
In particular, we present an algorithm that solves \textsc{\(\Delta\)-TVC} exactly, in time that is 
exponential in the number of edges of the underlying graph.
We then use this algorithm to give a \((d-1)\)-factor approximation algorithm (where \(d\) is the maximum vertex degree in any time-step) and finally give an \FPT-algorithm parameterized by the size of a solution.


Before we present our algorithms let us define the generalized notion of (sub)instances of our $\Delta$-TVC problem. We use this notion in the formulation of the exact exponential time algorithm and in the approximation algorithm.

\begin{definition}[\textsc{Partial \(\Delta\)-TVC}]
	Let \((G, \lambda)\) be a temporal graph.
	An instance of \textsc{Partial \(\Delta\)-TVC} is given by \((G,\lambda,\ell,h)\), where \(\ell: E(G) \to [T]\) and \(h: E(G) \to [T]\) map each edge to the starting time of its \emph{lowest uncovered time window} and \emph{highest uncovered time window} respectively. 
	The task is to find a temporal vertex subset \(\mathcal{C}\) of minimum size, 
	such that for every edge \(e \) and every time window \(W_i\) with \(\ell(e) \leq i \leq h(e)\) if \(e \in E[W_i]\) then there is some \((v,t) \in \mathcal{C}[W_i]\) where \(v \in e\), \(t \in \lambda(e)\).
\end{definition}
Intuitively, \textsc{Partial \(\Delta\)-TVC} of \((G,\lambda,\ell,h)\) is a $\Delta$-TVC of $(G,\lambda)$, where each edge $e \in E(G)$ has to be covered only in time windows $W_{l(e)}, W_{l(e)+1}, \dots, W_{h(e)}$.
Obviously, if \(\ell(e) = 1\) and \(h(e) = T - \Delta + 1\) for all \(e \in E(G)\) then
a solution of
\textsc{Partial \(\Delta\)-TVC} on \((G,\lambda,\ell,h)\) is also a valid solution of $\Delta$-TVC on $(G, \lambda)$.

\subsection{Exact Algorithm}
\label{subsec:exact}
\newcommand{\unddeg}{{d_G}}
In the following we denote by \(\unddeg\) the degree of the underlying graph $G$ of the considered instances \((G,\lambda,\ell,h)\) of \textsc{Partial \(\Delta\)-TVC}.
We provide a dynamic programming algorithm with running time $\mathcal{O}(T \Delta^{\mathcal{O}(|E(G)|)})$, as the next theorem states.
\begin{theorem}\label{thm:optPartialTVC}
	For every \(\Delta\geq 2\), a solution to \textsc{Partial \(\Delta\)-TVC} is computed in \(\mathcal{O}(Tc^{\mathcal{O}(|E(G)|)})\)  time, 
	where \(c = \min \{2^\unddeg,\Delta\}\) and \(T\) is the lifetime of the temporal graph in the instance.
\end{theorem}
Intuitively, our algorithm scans the temporal graph $\Delta$ time-steps at the time and for each edge, $e_i$ calculates the optimal vertex appearance that covers it.
Since a time-edge $(e_i, t)$ can appear at each time with a different combination of time-edges incident to it (we call this a \emph{configuration} of edges), it is not always the best to simply cover the last appearance of $e_i$ in the observed time window.
To overcome this property we show that it is enough to consider the latest appearance of each configuration of $e_i$, which drastically reduces the search space of our algorithm.
Let us now proceed with the detailed proof of \cref{thm:optPartialTVC}.
\iflong
\begin{proof}

Let \((G,\lambda,\ell,h)\) be an instance of \textsc{Partial \(\Delta\)-TVC}. We denote with $m$ the number of edges in the underlying graph $G$
and with $T$ the lifetime of the input instance, more precisely $T \in W_t$, where $t = \max_{e \in E(G)}(h(e))$.
For \((G,\lambda,\ell,h)\)  we define a dynamic programming table $f$, that
is indexed by tuples \((t,x_1,\dotsc, x_{m})\), where \(t \in [T - \Delta + 1]\) and \(x_i \in [0,\Delta]\) for all \(i \in [m]\).
We store in \(f(t, x_1, \dotsc, x_{m})\) the size of a minimum \textsc{Partial \(\Delta\)-TVC}
of a sub-instance $(G,\lambda,\ell',h)$ with $\ell'(e_i) = t + x_i$,
\ie we modify the lowest uncovered time window of each edge.
We call such a minimum solution a \emph{witness} for \(f(t,x_1, \dotsc, x_{m})\).
Whenever no solution exists we keep \(f(t,x_1, \dotsc, x_{m})\) empty.

The following observation is an immediate consequence of the definition and the fact that \(t + x_i = (t + 1) + (x_i - 1)\).
\begin{observation}\label{obs-f:increaset}
		If for all $i$, $x_i \geq 1$ then:
	$$f(t, x_1, x_2, \dots , x_m) = f(t+1, x_1 - 1, x_2 - 1, \dots , x_m - 1).$$
\end{observation}
We introduce the notion of \emph{configurations} of edges, 
that we use to recursively fill the dynamic programming table.
Let $v,w \in V(G)$ be the endpoints of edge $e_i$.
An edge configuration of time-edge $(e_i,t)$ at endpoint $v$ (resp.~$w$) consist of all time-edges that are incident to $(v,t)$ (resp.~$(w,t)$).
Formally, $\gamma(e_i,t)_v = \{e \in E(G) \mid v \in e_i \cap e, t \in \lambda(e)\}$ is a configuration of edges incident to time-edge $(e_i,t) $ at endpoint $v$.
With $\gamma (e_i)_v = \{\gamma(e_i,t)_v \mid t \in [T], t \in \lambda(e_i)\}$ we denote the \emph{set of all configurations} of edge $e_i \in (G,\lambda)$ at endpoint $v$.
Note, since each vertex in $G$ is of degree at most $d_G$ it follows that $e_i$ is incident to at most $2(d_G -1)$ edges in $G$ ($d_G -1$ at each endpoint) and therefore appears in at most $2^{d_G}$ different edge configurations, i.~e. $|\gamma(e_i)_{v_i}| + |\gamma(e_i)_{w_i}| \leq 2^{d_G}$ for all $v_iw_i=e_i \in E(G)$.
The applicability of edge configurations is presented by the following claim.

\begin{claim}
	Suppose that the edge $vw = e_i$ is optimally (temporally) covered in every time window up to $W_t$.
	Then,
	to determine an optimum temporal vertex cover of $e_i$ in the time window $W_t$, it is enough to check only the last appearance of every edge configuration of $e_i$ in $W_t$.
\end{claim}
\begin{proof}
	Suppose that 
	the configuration of $e_i$ at time $t_1$ is the same as the configuration of $e_i$ at time $t_2$ (i.e. $\gamma(e_i,t_1)_v = \gamma(e_i,t_2)_v$ and $\gamma(e_i,t_1)_w = \gamma(e_i,t_2)_w$), where $t \leq t_1 < t_2 \leq t + \Delta -1$.
	Let $\mathcal{C}$ be an optimum temporal vertex cover of $\mathcal{G}$ that covers edge $e_i$ at time $t_1$ with the vertex appearance $(v,t_1)$.
	Now let $\mathcal{C}' = \{ \mathcal{C} \setminus (v, t_1)\} \cup (v, t_2)$.
	Since $\mathcal{C}'$ and $\mathcal{C}$ differ only in the time window $W_t$
	all edges in configuration $\gamma(e_i,t_1)_v$ remain optimally (temporally) covered in all time windows up to $W_t$.
	Since $(v,t_2)$ temporally covers the same edges as $(v,t_1)$ in the time window $W_t$, it follows that $\mathcal{C}'$ is also a $\Delta$-TVC.
	Clearly $\mathcal{C}'$ is not bigger than $\mathcal{C}$.
	Therefore, $\mathcal{C}'$ is also an optimum solution.
\end{proof}

\noindent
Using the properties presented above we are ready to describe the procedure that computes the value of $f(t,x_1,x_2, \dots, x_m)$.
\begin{itemize}
	\item[Case 1.] If for all $i$ it holds $x_i > 0$,
	then $f(t, x_1, x_2, \dots, x_m) = f(t+1, x_1-1, x_2-1, \dots, x_m-1)$.
	\item[Case 2.] There is at least one $x_i = 0$.
	Let $i$ be the smallest index for which this holds and let $v_iw_i=e_i$ be its corresponding edge.
	Now, to temporally cover $e_i$ in the time window $W_t$ we check the last appearance of all edge configurations of $e_i$ at $v_i, w_i$ in the time window $W_t$, and choose the one that minimizes the
	solution. 
	Namely, if we cover $e_i$ at time $t_i$ using vertex $v_i$ we add $(v_i,t_i)$ to the solution, set the new value of $x_i$ to $t_i - t$
	and also the value $x_j$ of every edge $e_j$ in $\gamma(e_i,t)_{v_i}$ to $t_i - t + 1$.
	We then recursively continue the calculation of $f$.
	At the end, we select such $t_i$ (such edge configuration) that minimizes the value $f(t,x_1,x_2,\dots,x_m)$.
\end{itemize}

\paragraph{Running Time.}
It is straightforward to see that the number of table entries is bounded by \(T\Delta^{|E(G)|}\) which also bounds the complexity of computing them.\\
A second bound can be established by noticing that each entry of \(f\), that is filled during the dynamic program, is uniquely determined by configurations of each edge.
Together both bounds yield the desired complexity.
\end{proof}
\fi

The above algorithm solves \textsc{Partial \(\Delta\)-TVC} in \(\mathcal{O}(Tc^{\mathcal{O}(|E(G)|)})\), where \(c = \min \{2^{d_\Delta},\Delta\}\).
Note, that the following parameterized complexity result 
of \dTVC\ follows directly from our algorithm above.

\begin{corollary}
	For every $\Delta\geq 2$, \dTVC\ can be solved in time in \(\mathcal{O}(Tc^{|E(G)|})\) where \(c = \min \{2^{\mathcal{O}(|E(G)|)},\Delta\}\), 
	and thus \dTVC\ is in \FPT\ parameterized by \(|E(G)|\).
\end{corollary}

\subsection{\boldmath Approximation Ratio Better Than \(d\)}\label{better-d-approx-subsec}
In this section, we focus on approximation algorithms.
In \cite{AkridaMSZ20}, the authors present the \(d\)-factor approximation algorithm (where \(d\) is the maximum vertex degree in any time-step), which uses the following idea.
First optimally solve the $\Delta$-TVC for each edge separately, which can be done in polynomial time, and then combine all the solutions.
In the worst case, the solution may include all $d$ neighbours of a specific vertex appearance $(v,t)$, instead of just $(v,t)$, which is how the bound is obtained.
We try to avoid this situation by iteratively covering paths with two edges (\(P_3\)-s) instead of single edges. With this approach, we force the middle vertex (called \emph{central vertex}) to be in the temporal vertex cover,
which gives us an error of at most \(d-1\).
Based on this idea we can formulate our \((d-1)\)-approximation algorithm. Note in our case we require that $d \geq 3$.
\subsubsection{Description of the Algorithm}
Let $\mathcal{G} = (G,\lambda)$ be an instance of $\Delta$-TVC.
We iteratively extend an initially empty set \(\mathcal{X}\) to a $\Delta$-TVC of $\mathcal{G}$ in the following way.
While there is an edge \(e \in E(G)\) that is not covered in some time window of \(\mathcal{G}\) we have to extend \(\mathcal{X}\); otherwise \(\mathcal{X}\) is a \(\Delta\)-TVC which we can return.
We proceed in two phases:

\iflong
\begin{enumerate}
\item 
While there is a 3-vertex path \(P\) that appears at some time step \(t\) and whose underlying edges are uncovered by \(\mathcal{X}\) in some time window $W_i$ with \(t \in W_{i}\),
we build an instance of the \textsc{Partial \(\Delta\)-TVC} in the following way.
Let \(S\) be the set of all time-steps $t$ in which both edges of $P$ appear and for which there are time windows containing that time-step in which each edge of \(P\) is uncovered.
More specifically,
\[S = \{t \in [T] \mid \begin{aligned}[t]
	& P \text{ appears at } t \text{ and}\\
	& \text{for every } e \in E(P) \text{ there is a } W_i \text{ s.t.\ } t \in W_i \text{ and } e \text{ is uncovered in } W_i\}.
\end{aligned}\]

We subdivide \(S\) into subsets \(S_1, \dotsc S_k\) with \(k \leq T\) such that
there is a gap of at least $\Delta$ between the last element of $S_i$ and the first element of $S_{i+1}$ (for all $i \in [k-1]$).
Now we iteratively consider each \(S_i\) and construct the \textsc{Partial $\Delta$-TVC} instance $\mathcal{H}$ given by \((P, {\lambda'}^P_i, \ell_i(P), h_i(P))\)
with 
\({\lambda'}^P_i\) defined as a mapping of each edge of \(P\) to \(S^P_i\),
\(\ell_i(P)\) is the smallest value \(t\) such that \((P,\min S^P_i)\) is not covered in time window \(W_t\) by \(\mathcal{X}\),
and \(h_i(P)\) is the smallest value \(t\) such that \((P,\max S^P_i)\) is not covered in time window \(W_t\) by \(\mathcal{X}\).
We use the algorithm from \Cref{subsec:exact} to solve such an instance, extend \(\mathcal{X}\) by the union of this solution, and continue with \(S^P_{i+1}\).\\
Note that these are the instances that bring the crucial advantage compared to the known \(d\)-approximation.
To ensure our approximation guarantee, we also need to allow the inclusion of appearances of a single underlying edge in different \(P_3\) instances for an example see \cref{fig:approximation-multipleEdges}.
This is why we restrict the temporal extent of each \(P_3\) instance to time-steps which are impacted by the actual appearances of the respective \(P_3\).
\begin{figure}[t]
	\centering
	\includegraphics[width=0.25\linewidth]{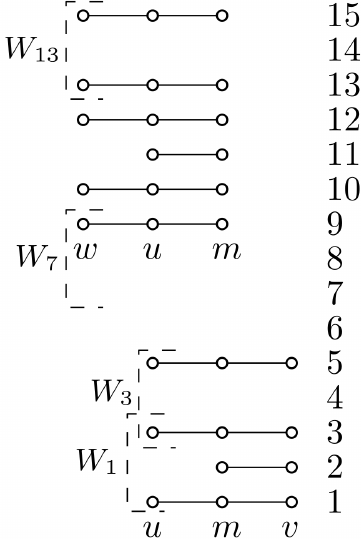}
	\caption{An example of a temporal graph with $\Delta = 3$ and two of the \textsc{Partial \(\Delta\)-TVC} instances $\mathcal{H}_1 = (P_1, \lambda_1, \ell_1, h_1)$ and $\mathcal{H}_2 = (P_2, \lambda_2, \ell_2, h_2)$, generated by our algorithm.
	Here, $\mathcal{H}_1$~is defined as $P_1 = (u, m, v)$, $S_1 = \{1, 3, 5\}$, $\ell_1 = 1$, $h_1 = 3$, $\lambda_1(um) = \lambda_1(mv) = S_1$.
	Also, $\mathcal{H}_2$ is defined as $P_2 = (w, u, m)$, $S_2 = \{9, 10, 12, 13, 15\}$, $\ell_2 = 7$, $h_2 = 13$, $\lambda_2(wu) = \lambda_2(um) = S_2$.
	Note that $\mathcal{H}_1 \cap \mathcal{H}_2 = \emptyset$, i.e.~$\mathcal{H}_1$ and $\mathcal{H}_2$ are two disjoint \textsc{Partial \(\Delta\)-TVC} instances, while the edge $um$ is present in both underlying paths $P_1$ and $P_2$.
	\label{fig:approximation-multipleEdges}}
\end{figure}

\item 
When there is no such $P_3$ left, we set \(F\) to be the set of edges \(e \in E(G)\) that are not yet covered in some time window of \(\mathcal{G}\).
For \(e \in F\), let \(S^e\) be the set of all time-steps in which \(e\) appears and is not covered.
We subdivide \(S^e\) into subsets \(S^e_1, \dotsc S^e_k\) with \(k \leq T\) such that
there is a gap of at least $2\Delta - 1$ between the last element of $S^e_i$ and the first element of $S^e_{i+1}$ (for all $i \in [k-1]$).
Now we consider the \textsc{Partial $\Delta$-TVC} instance given by \((\{e\}, {\lambda'}^e_i, \ell_i(e), h_i(e))\)
with \({\lambda'}^e_i\) defined a mapping of $e$ to \(S^e_i\),
\(\ell_i(e)\) is the smallest value \(t\) such that \((e,\min S^e_i)\) is not covered in time window \(W_t\),
and \(h_i(e)\) is the largest value \(t\) such that \((e,\max S^e_i)\) is not covered in time window \(W_t\).
We use the algorithm from \Cref{subsec:exact} to solve these instances and extend \(\mathcal{X}\) by the union of the solutions.\\
Intuitively, these instances just take care of whatever remains after the first phase.
We split them up temporally, just for compatibility with the input for the algorithm from \Cref{subsec:exact}.
\end{enumerate}
\fi
\ifshort
\smallskip
\noindent\textbf{Phase 1}: While we find two such edges \(e_1,e_2\) that are adjacent and appear at the same time-step and both appearances are not covered in some time window \(W_i\) then we consider the following instances of \textsc{Partial \(\Delta\)-TVC}:
Let \(S\) be the set of all time-steps in which \(e_1\) and \(e_2\) appear together and are both not covered by \(\mathcal{X}\) in some time window \(W_i\) with \(i \in [\min S - \Delta + 1, \max S - \Delta + 1]\).
We can subdivide \(S\) into subsets \(S_1, \dotsc S_k\) with \(k \leq T\) such that \(S_1\) contains the smallest elements of \(S\) such that there is no gap of at least \(2\Delta - 1\) between its elements, \(S_2\) contains the smallest elements of \(S\) between the first and second gap of at least \(2\Delta - 1\), and so on. 
Now we consider the subinstances given by \((G[e_1,e_2], \lambda', \ell, h, \min S_i, \max S_i)\)
with \(\lambda'\) defined as the restriction of \(\lambda\) to \(\{e_1,e_2\} \cap \lambda^{-1}([\min S_i - \Delta + 1, \max S_i - \Delta + 1])\),
\(\ell(e_1) = \ell(e_2)\) is the smallest time-step \(t\) such that \((e_1,\min S_i)\) and \((e_2, \min S_i)\) are not covered in time window \(W_t\) by \(\mathcal{X}\),
and \(h(e_1) = h(e_2)\) is the largest time-step \(t\) such that \((e_1,\max S_i)\) and \((e_2, \max S_i)\) are not covered in time window \(W_t\) by \(\mathcal{X}\).
We use the algorithm from \Cref{subsec:exact} to solve these subinstances and extend \(\mathcal{X}\) by the union of the solutions.

\smallskip
\noindent\textbf{Phase 2:} If no such edges are adjacent and appear at the same time-step and both appearances are not covered in some time window \(W_i\) let \(F\) be the set of edges \(e \in E(G)\) with occurrences which is not covered in some time window in the lifetime of \(\mathcal{G}\) by \(\mathcal{X}\).
We consider the following subinstances of \textsc{Partial \(\Delta\)-TVC}:
For \(e \in F\), let \(S^e\) be the set of all time-steps in which \(e\) appears and is not covered by \(\mathcal{X}\) in some time window \(W_i\) with \(i \in [\min S^e - \Delta + 1, \max S^e - \Delta + 1]\) and \(t \in S^e\).
We can subdivide \(S^e\) into subsets \(S^e_1, \dotsc S^e_k\) with \(k \leq T\) such that \(S^e_1\) contains the smallest elements of \(S^e\) such that there is no gap of at least \(2\Delta - 1\) between its elements, \(S^e_2\) contains the smallest elements of \(S^e\) between the first and second gap of at least \(2\Delta -1\), and so on. 
Now we consider the subinstances given by \((G[e], \lambda', \ell, h, \min S^e_i, \max S^e_i)\)
with \(\lambda'\) defined as the restriction of \(\lambda\) to \(\{e\} \cap \lambda^{-1}([\min S^e_i - \Delta + 1, \max S^e_i + \Delta])\),
\(\ell(e)\) is the smallest time-step \(t\) such that \((e,\min S^e_i)\) is not covered in time window \(W_t\) by \(\mathcal{X}\),
and \(h(e)\) is the largest time-step \(t\) such that \((e,\max S^e_i)\) is not covered in time window \(W_t\) by \(\mathcal{X}\).
We use the algorithm from \Cref{subsec:exact} to solve these subinstances and extend \(\mathcal{X}\) by the union of the solutions.
\fi

It follows from the above construction that the produced set \(\mathcal{X}\) of vertex appearances is a \(\Delta\)-TVC of $\mathcal{G}$.

\ifshort
Furthermore, using a double-counting argument, we can show that the approximation ratio is at most $d-1$.
\fi

\subsubsection{Running Time and the Approximation Ratio \boldmath \((d-1)\).}
{
	The number of instances considered in Phase~1 is bounded by the number of combinations of two edges in \(E(G)\) multiplied by \(T\).
	Similarly, the number of instances considered in Phase~2 is bounded by the number of edges in \(E(G)\) multiplied by \(T\).
	To find an optimal solution for each instance we use the algorithm from \cref{thm:optPartialTVC}. Since each instance has at most $2$ edges and the maximum degree is at most $2$, 
	the algorithm requires \(\mathcal{O}(T)\) time to solve it optimally.
	Therefore, the overall running time lies in \(\mathcal{O}(|E(G)|^2T^2)\).
	
	Now we argue the claimed approximation guarantee.
	For this we denote with \(\mathfrak{H}\) the set of all  \textsc{Partial \(\Delta\)-TVC} 
	instances created by our algorithm.
	Let \(\mathcal{C}\) be a minimum \(\Delta\)-TVC for \(\mathcal{G}\), and \(\mathcal{X}\) be the temporal vertex set computed by our algorithm.
	A  \textsc{Partial \(\Delta\)-TVC} instance $\mathcal{H} \in \mathfrak{H}$ is of form $\mathcal{H} = (H,\lambda'_H,\ell_H,h_H)$, where $H$ is either a $P_3$ (Phase~1 instance) or an edge (Phase~2 instance).
	
	Let us now define the set $A$ as the set of triples $(v,t,\mathcal{H})$, 
	where $(v,t) \in \mathcal{C}$ is a vertex appearance from the optimum solution $\mathcal{C}$, 
	and $\mathcal{H}= (H,\lambda'_H,\ell_H,h_H) \in \mathfrak{H}$ is such \textsc{Partial $\Delta$-TVC} instance, created by our algorithm,
	that the vertex appearance $(v,t)$ covers one appearance of an edge of $H$ 
	in a time window between $W_{\ell_H}$ and $W_{h_H}$. 
	More precisely,
	\begin{align*}
	A = \{(v,t,\mathcal{H}) \mid & (v,t) \in \mathcal{C}, \mathcal{H} \in \mathfrak{H} \, \land\\
	& (v,t) \text{ covers an appearance of } e \in E(H)\\
	& \text{at some time } t' \in \lambda(e)\\
	& \text{in } W_i \text{ with } \ell_H(e) \leq i \leq h_H(e)\}\end{align*}
	From the definition of $\mathcal{X}$ and $A$ it follows that
	\begin{equation}\label{eq:d-1approx-bound1}
		|\mathcal{X}| \leq |A|,
	\end{equation} 
	as our algorithm optimally covers every \textsc{Partial $\Delta$-TVC} instance $\mathcal{H} \in \mathfrak{H}$,
	and, in case vertex appearances in the optimal solution $(v,t) \in \mathcal{C}$ differ from the ones chosen by our algorithm, we may need more of them to cover each $\mathcal{H}$.
	
	To show the desired approximation guarantee we give a double-counting argument for the cardinality of $A$.
	Counting from one side we distinguish vertices in the optimal solution $\mathcal{C}$ based on what kind of \textsc{Partial $\Delta$-TVC} instances $\mathcal{H} \in \mathfrak{H}$
	they cover time-edges in.
	To count the number of elements in $A$ from the other side we distinguish the \textsc{Partial $\Delta$-TVC} instances $\mathcal{H} \in \mathfrak{H}$ 
	by the number of vertices in a fixed optimal solution $\mathcal{C}$, that cover the edges appearing in $\mathcal{H}$.
	We either use the optimal amount of vertices in $\mathcal{C}$ to cover $\mathcal{H}$ or we have to use more.
	
	More specifically, on one side we want to associate \textsc{Partial $\Delta$-TVC} instances $\mathcal{H}$ with vertex appearances $(v,t) \in \mathcal{C}$. For this, we define the following set
	\begin{align*}
	\mathfrak{H}(v,t) = \{\mathcal{H} \in \mathfrak{H} \mid & (v,t) \in \mathcal{C} \text{ covers an appearance}\\
	& \text{of } e \in E(H) \text{ at time } t' \in \lambda(e)\\
	&\text{in } W_i \text{ with } \ell_H(e) \leq i \leq h_H(e)\}.
	\end{align*} 
	On the other side, we want to associate vertex appearances $(v,t) \in \mathcal{C}$ with the  \textsc{Partial $\Delta$-TVC} instances $\mathcal{H}$. For this, we define the following set
	\begin{align*}
	\mathcal{C}(\mathcal{H}) = \{(v,t) \in \mathcal{C} \mid & (v,t) \text{ covers an appearance}\\
	& \text{of } e \in E(H) \text{ in } W_i \text{ with }\\
	& \ell_H(e) \leq i \leq h_H(e)\}.
	\end{align*}
	Intuitively, one can think of the sets of form $\mathfrak{H}(v,t)$, for some vertex appearance $(v,t)$, as the collection of \textsc{Partial $\Delta$-TVC} instances ``touching'' $(v,t)$,
	and $\mathcal{C}(\mathcal{H})$, for some \textsc{Partial $\Delta$-TVC} instance $\mathcal{H}$, as the set of vertex appearances ``touching'' $\mathcal{H}$. 
	For an example see \cref{fig:d-1ApproximationSets}.
	\begin{figure}[h]
		\centering
		\includegraphics[width=0.3\linewidth]{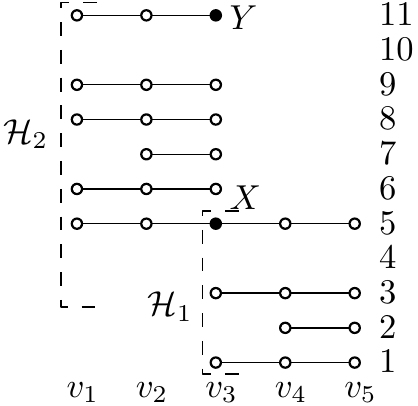}
		\caption{An example of a temporal graph with $\Delta = 3$ and two of the \textsc{Partial \(\Delta\)-TVC} instances $\mathcal{H}_1 = (P_1, \lambda_1, \ell_1, h_1)$ and $\mathcal{H}_2 = (P_2, \lambda_2, \ell_2, h_2)$, generated by our algorithm.
			Here, $\mathcal{H}_1$~is defined as $P_1 = (v_3, v_4, v_5)$, $S_1 = \{1, 3, 5\}$, $\ell_1 = 1$, $h_1 = 3$, $\lambda_1(v_3 v_4) = \lambda_1(v_4 v_5) = S_1$.
			Also, $\mathcal{H}_2$ is defined as $P_2 = (v_1, v_2, v_3)$, $S_2 = \{5, 6, 8, 9, 11\}$, $\ell_2 = 3$, $h_2 = 9$, $\lambda_2(v_1 v_2) = \lambda_2(v_2 v_3) = S_2$.
			We have two vertex appearances $X = (v_3, 5)$ and $Y=(v_3,11)$ from $\mathcal{C}$, and
		their corresponding sets 
		$\mathfrak{H}(X) = \{ \mathcal{H}_1, \mathcal{H}_2\}$,
		$\mathfrak{H}(Y) = \{ \mathcal{H}_2\}$,
		$\mathcal{C}(\mathcal{H}_1) = \{ X\}$,
		$\mathcal{C}(\mathcal{H}_2) = \{ X, Y\}$.}
		\label{fig:d-1ApproximationSets}
	\end{figure}

	For brevity, when speaking of a \textsc{Partial $\Delta$-TVC} instance \(\mathcal{H}\) constructed by our algorithm and a vertex appearance \((v,t)\), we say that \((v,t)\) \emph{subcovers} an edge \(e \in E(H)\) if \((v,t)\) covers an appearance of  \(e\) in \(W_i\) with \(\ell_H(e) \leq i \leq h_H(e)\).
 	
	We can now begin with counting the number of elements in $A$.
	From the definition, it follows that
	\begin{equation*}
	|A|  = \sum_{(v,t) \in \mathcal{C}} |\mathfrak{H}(v,t)|.
	\end{equation*}
	We can split the sum into two cases, 
	first where one vertex appearance subcovers one edge (or even a $P_3$) of at most $d-1$ different \textsc{Partial $\Delta$-TVC} instances,
	and the second when one vertex appearance subcovers one edge (or a $P_3$) of more than $d-1$ different \textsc{Partial $\Delta$-TVC} instances.
	\begin{equation} \label{eq:Counting-1}
	|A| = \sum_{\substack{(v, t) \in \mathcal{C} \\ |\mathfrak{H}(v,t)| \leq d - 1}} |\mathfrak{H}(v,t)| + \sum_{\substack{(v, t) \in \mathcal{C} \\ |\mathfrak{H}(v,t)| > d - 1}} |\mathfrak{H}(v,t)| 
	\end{equation}
	By the construction of the \textsc{Partial $\Delta$-TVC} it follows that if $(v,t)$ has a temporal degree $k$, it subcovers edges in at most $k$ different \textsc{Partial $\Delta$-TVC} instances.
	Moreover, if a vertex appearance $(v,t)$ subcovers two edges of a single \textsc{Partial $\Delta$-TVC} instance $\mathcal{H}$ 
	(\ie $\mathcal{H}$ is a Phase~1 instance with a $P_3$ as the underlying graph and $v$ is a central vertex of $P_3$),
	then $(v,t)$ can subcover time-edges in one less \textsc{Partial $\Delta$-TVC} instance.
	More specifically, if we assume that $(v,t)$ with temporal degree \(k\) is a central vertex of at least one \textsc{Partial $\Delta$-TVC} instance (\ie subcovers both edges of at least one instance),
	then $(v,t)$ can subcover edges in at most $k-1$ different instances.
	Note that, since the degree of each vertex is at most $d$, by our assumption,
	this means that for all vertex appearances $(v,t) \in \mathcal{C}$ that subcover both time-edges of some Phase~1 \textsc{Partial $\Delta$-TVC} instance
	it holds that \(|\mathfrak{H}(v,t)| \leq d - 1\).
	Therefore, for the second sum of \cref{eq:Counting-1} to happen vertex appearance $(v,t) \in \mathcal{C}$ must subcover exactly one underlying edge of $d$ different
	\textsc{Partial $\Delta$-TVC} instances.
	\newcommand{\criticalset}{D}
	For brevity we denote \(\criticalset := \{(v,t) \in \mathcal{C} \mid |\mathfrak{H}(v,t)| = d\}\)
	Hence the following holds
	\begin{align}
	\nonumber
	|A|& \leq (d - 1)(|\mathcal{C}|-|\criticalset|) + \sum_{(v, t) \in \criticalset} |\mathfrak{H}(v,t)|\\
	\nonumber
	& \leq (d - 1)|\mathcal{C}| + \sum_{(v, t) \in \criticalset} (d - (d - 1))\\
	& = (d - 1)|\mathcal{C}| + |D|. \label{eq:d-1approx-bound2}
	\end{align}
	Let us now count the number of elements in $A$ from the other side.
	Again by the definition, we get that
	\begin{equation*}
	|A| = \sum_{\mathcal{H} \in \mathfrak{H}} |\mathcal{C}(\mathcal{H})|.
	\end{equation*}
	Since $|A| \geq |\mathcal{X}|$ (see~\cref{eq:d-1approx-bound1}), we can divide $A$ into two parts. The first one, where the solution for the partial instances computed by our algorithm coincide (in terms of cardinality) with the ones implied by $\mathcal{C}$, and the second where \(\mathcal{C}\) implies suboptimal solutions. 
	For this we define
	\(\mathcal{X}(\mathcal{H})\) as the optimal solution for \textsc{Partial $\Delta$-TVC} instance \(\mathcal{H}\) calculated by our algorithm.
	From the construction it holds that \(|\mathcal{X}| = \sum_{\mathcal{H} \in \mathfrak{H}}|\mathcal{X}(\mathcal{H})|\).
	We get the following
	\begin{equation}\label{eq:d-1approx-bound3}
	|A| = |\mathcal{X}| + \sum_{\mathcal{H} \in \mathfrak{H}} 
	\left(
	|\mathcal{C}(\mathcal{H})| - |\mathcal{X}(\mathcal{H})|
	\right).
	\end{equation}
	
	Combining \cref{eq:d-1approx-bound2,eq:d-1approx-bound3} we obtain
	\begin{equation*}
		|\mathcal{X}| + \sum_{\mathcal{H} \in \mathfrak{H}} 
		\left(
		|\mathcal{C}(\mathcal{H})| - |\mathcal{X}(\mathcal{H})|
		\right)
		\leq
		(d - 1)|\mathcal{C}| + |D|.
	\end{equation*}
	Now, to prove our target, namely that
	$|\mathcal{X}| \leq (d-1)|\mathcal{C}|$, 
	it suffices to show that
	\begin{equation}\label{eq:d-1approx-main}
	|\criticalset|
	\leq \sum_{\mathcal{H} \in \mathfrak{H}} 
	\left(
	|\mathcal{C}(\mathcal{H})| - |\mathcal{X}(\mathcal{H})|
	\right).
	\end{equation}
	Recall that $|\mathfrak{H}(v,t)| > (d - 1)$ holds only for vertex appearances $(v,t) \in \mathcal{C}$, 
	with temporal degree $d$, that cover appearances of exactly one edge from every \textsc{Partial $\Delta$-TVC} $\mathcal{H} \in \mathfrak{H}(v,t)$.
	Furthermore, it turns out that at most one \textsc{Partial $\Delta$-TVC} $\mathcal{H} \in \mathfrak{H}(v,t)$ is a Phase~2 (a single edge) instance.
	\begin{lemma}
		\label{obs:covered-instances}
		Let \((v,t) \in \mathcal{C}\).
		There is at most one Phase~2 (single edge) \textsc{Partial $\Delta$-TVC} instance \(\mathcal{H}_1 \in \mathfrak{H}(v,t)\).
	\end{lemma}
	\begin{proof}
		Suppose this is not true. 
		Then there are two different time-edges $(e,t), (f,t)$ incident to $(v,t)$, that belong to Phase~2 \textsc{Partial $\Delta$-TVC} instances $\mathcal{H}_1$ and $\mathcal{H}_2$, respectively.
		W.l.o.g.~suppose that our algorithm created instance $\mathcal{H}_1$ before $\mathcal{H}_2$.
		By the construction algorithm creates a Phase~2  \textsc{Partial $\Delta$-TVC} instance whenever there is no uncovered $P_3$ left in the graph. 
		But, when creating $\mathcal{H}_1$ both time-edges $(e,t), (f,t)$ are uncovered in some time windows \(W_e\) and \(W_f\), respectively, where \(t \in W_e \cap W_f\).
		Therefore $(e,t), (f,t)$ form an uncovered $P_3$ path at time-step $t$,
		which would result in algorithm combining them into one single Phase~1  \textsc{Partial $\Delta$-TVC} instance.
	\end{proof}
		We now proceed with the proof of \cref{eq:d-1approx-main}.
		\begin{lemma}
			\label{lem:critical-subinstances}
			For $\criticalset$ and $\sum_{\mathcal{H} \in \mathfrak{H}} \left(
			|\mathcal{C}(\mathcal{H})| - |\mathcal{X}(\mathcal{H})|
			\right)$ 
			defined as above the following holds
			\begin{equation}
				|\criticalset|
				\leq \sum_{\mathcal{H} \in \mathfrak{H}} 
				\left(
				|\mathcal{C}(\mathcal{H})| - |\mathcal{X}(\mathcal{H})|
				\right).
		\end{equation}\end{lemma}
		\begin{proof}
			We show the statement by modifying the solutions implied by $\mathcal{C}$ on Phase~1 \textsc{partial $\Delta$-TVC} instances (paths of length $3$).
			More specifically, for each \((v,t) \in \criticalset\) and every Phase~1 \textsc{partial $\Delta$-TVC} instance \(\mathcal{H} \in \mathfrak{H}(v,t)\), we find a vertex appearance in \(\mathcal{C}(\mathcal{H})\) which is ``disposable in the solution''.
			We refer to a set of vertex appearances including $(v,t)$ as disposable in the solution of $\mathcal{H} \in \mathfrak{H}(v,t)$ if we can modify the solution \(\mathcal{C}(\mathcal{H})\) by decreasing its size by the size of that set.
			In our process we not only show that $(v,t)$ is disposable in the solution of $\mathcal{H}$,
			but that the set of all vertex appearances of $v$ in $\mathcal{C}(\mathcal{H})$ (meaning all vertex appearances of form $(v,t_i)$)
			are disposable in the solution of $\mathcal{H}$.
			To produce a solution \(\widetilde{\mathcal{C}(\mathcal{H})}\) of $\mathcal{H}$ which proves that certain vertex appearances are disposable in the solution \(\mathcal{C}(\mathcal{H})\) of $\mathcal{H}$ we perform the following three steps.
			Firstly, we modify the solution
			\(\mathcal{C}(\mathcal{H})\) to only include vertex appearances at time-steps when the entire \(P_3\) underlying \(\mathcal{H}\) appears.
			Secondly, we ensure that the included vertex appearances corresponds to the appearances of the central vertex of the \(P_3\) underlying \(\mathcal{H}\).
			Thirdly, we remove unnecessary vertex appearances, producing the desired solution.
			The process is depicted in \cref{fig:approx-bound-proof}.
			\begin{figure}
				\centering
				\includegraphics[width=0.2\linewidth]{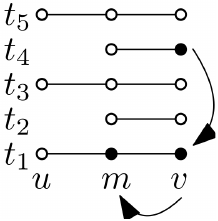}
				\caption{An illustration of a modification of a solution $\mathcal{C}(\mathcal{H})$ presented in \cref{lem:critical-subinstances}.
				\\
				The figure represents a \textsc{partial $\Delta$-TVC} instance $\mathcal{H}$, where the vertex appearance $(v,t_4) \in \mathcal{C}$.
				We first modify the solution $\mathcal{C}(\mathcal{H})$ by replacing $(v,t_4)$ with $(v,t_1)$, and then further replace it with the central vertex $(m,t_1)$.
				\label{fig:approx-bound-proof}}
			\end{figure}
			%
			While determining disposable vertex appearances does not correspond to an injection on \(\criticalset\) we ensure that no single disposable vertex appearance is used for more than two elements of \(\criticalset\), which proves to be sufficient in demonstrating the desired bound.
			
			We start by inspecting vertex appearances \((v,t) \in \criticalset\) in an order increasing with \(t\).
			Recall that \((v,t) \in \criticalset\) implies that there is no Phase~1 \textsc{Partial $\Delta$-TVC} instance in \(\mathfrak{H}(v,t)\) in which the associated temporal graph contains two edge appearances incident to \(v\) at \(t\), because otherwise \(|\mathfrak{H}(v,t)| \leq d - 1\).
			Since \(d - 1 \geq 2\), together with \Cref{obs:covered-instances} this implies that there are at least two Phase~1 \textsc{Partial $\Delta$-TVC} instances in \(\mathfrak{H}(v,t)\) for each of which \(v\) is incident to only one of the underlying edges.
			Notice that here it is important that \(d \geq 3\).

			Let $\mathcal{H}$ be a Phase~1 subinstance from $\mathfrak{H}(v,t)$.
			W.l.o.g.~we denote with $P = (u,m,v)$ the underlying path of $\mathcal{H}$,
			and let $\ell_P$ and $h_P$ be the lowest and the highest time-steps in the definition of $\mathcal{H}$, respectively.
			Note that vertex appearances in the intersection $\mathcal{C}(\mathcal{H}) \cap \criticalset$ are of the form $(v,t_i),(u,t_j)$ for some $t_i, t_j$ (if they exist).
			Suppose that there are $k_v \geq 1$ vertices of form $(v,t_i)$ in \(\mathcal{C}(\mathcal{H})\).
			Our goal is to create the modified solution $\widetilde{\mathcal{C}(\mathcal{H})}$ of $\mathcal{H}$
			where $|\widetilde{\mathcal{C}(\mathcal{H})}| \leq |\mathcal{C}(\mathcal{H})| - k_v$.
			We start by setting 
			$\widetilde{\mathcal{C}(\mathcal{H})} = \mathcal{C}(\mathcal{H})$.
			We first iterate through every vertex appearance $(u,t') \in \mathcal{C}(\mathcal{H})$ (increasingly in $t'$), and do the following.
			We set $T_\mathcal{H} = [\ell_P, h_P + \Delta - 1] \cap [t'-\Delta +1, t'+\Delta -1]$ to be the interval of time-steps $t''$ for which $\mathcal{H}$ is defined, such that $|t'' - t'| \leq \Delta$.
			By the construction of Phase~1 subinstances, both edges $um$ and $mv$ appear at some time step \(t^*\) in $T_\mathcal{H}$.
			We chose \(t^* \in T_\mathcal{H}\) to be maximal such that \((x,t^*)\) covers an appearance of the edge $um$ 
			in the earliest time window $t^* \in W_i$ of \(\mathcal{H}\),
			where the edge $um$ is not yet covered by vertex appearances in 
			$\widetilde{\mathcal{C}(\mathcal{H})} \setminus (u,t')$.
			Intuitively, when removing $(u,t')$ from $\widetilde{\mathcal{C}(\mathcal{H})}$ we may get some time-windows $W_i, W_{i+1}, W_{i+2} \dots$, where the edge $um$ is not yet covered. We now set $t^*$ to be the maximum value in $W_i$ when both $um$ and $mv$ appear. 
			If such a $(u,t^*)$ exists we replace $(u,t')$ by $(u,t^*)$, otherwise we identify $(u,t')$ as disposable in the solution for $\mathcal{H}$.
			%
			We then further modify \(\widetilde{\mathcal{C}(\mathcal{H})}\) by replacing each \((u,t^*)\) from the previous step with \((m,t^*)\).
			We now repeat the same procedure for all $(v,t') \in \widetilde{\mathcal{C}(\mathcal{H})}$.
			After the first two steps are finished $\widetilde{\mathcal{C}(\mathcal{H})}$ consists only of the vertex appearances of the central vertex $m$ of $P$.
			We now continue with the final step and iterate through
			the vertex appearances \((m,t)\) in the order of increasing \(t\), delete unnecessary vertices from the solution $\widetilde{\mathcal{C}(\mathcal{H})}$ for \(\mathcal{H}\) and mark them as disposable.
			This happens if
			\((m,t)\) has been added to \(\widetilde{\mathcal{C}(\mathcal{H})}\) multiple times, when considering different vertex appearances,
			or it can be removed from \(\widetilde{\mathcal{C}(\mathcal{H})}\) without changing the fact that it is a solution.
			In both cases, any appearance of the vertex $v$ from the set \(\criticalset\) which is associated with the unnecessary vertex appearance \((m,t^*)\) is disposable in the solution for \(\mathcal{H}\).
			
			We now have to show that $\widetilde{\mathcal{C}(\mathcal{H})}$ is indeed a solution for $\mathcal{H}$.
			To prove this we have to make sure that exchanging each vertex appearance $(u,t') \in \mathcal{C}(\mathcal{H})$ (resp.~$(v,t') \in \mathcal{C}(\mathcal{H})$) with $(u,t^*)$ (resp.~$(v,t^*)$) and then finally with $(m,t^*)$ in the first two steps of the modification returns a valid solution.
			Suppose for the contradiction that the edge $mv$ appears at time $t$ and that there is a time window $t \in W_i$ where $mv$ is uncovered in $W_i$ by our final solution $\widetilde{\mathcal{C}(\mathcal{H})}$.
			Using the definition of $\mathcal{H}$ and the fact that $\widetilde{\mathcal{C}(\mathcal{H})}$ contains only the appearances of the central vertex of $P$
			it follows that also $um$ appears at time $t$ and is uncovered in $W_i$.
			In the initial solution, $um$ and $mv$ were both covered in $W_i$.
			If they were both covered by the same vertex appearance, it must have been the appearance of the central vertex which would remain in our final solution. Therefore, the edges are covered by two different vertex appearances. Denote them with $(u,t_u)$ and $(v,t_v)$, respectively.
			Note that $t_u, t_v \in W_i$. Therefore, $|t_v - t_u| \leq \Delta$.
			In our procedure we first considered $(u, t_u)$ that was replaced by $(u,t_u^*)$ and then by $(m,t_u^*)$, where the edge $um$ was uncovered by $\widetilde{\mathcal{C}(\mathcal{H})} \setminus (u, t_u)$ in a time window $W_j^u$.
			It also holds that $t_u^*$ is the maximum value of $W_j^u$ when both $um$ and $mv$ occur.
			Since $um$ and $mv$ are uncovered in $W_i$ it must be true also that $t_u^* \notin W_i$, even more $t_u^* < i$.
			Next, the procedure considered $(v,t_v)$, and replaced it by $(v,t_v^*)$ and then further by $(m,t_v^*)$.
			Note here that $t_v^*$ is the maximum value in some time window $W_j^v$ where both $um$ and $mv$ appear and are uncovered by the modified solution $\widetilde{\mathcal{C}(\mathcal{H})} \setminus (v, t_v)$, where $\widetilde{\mathcal{C}(\mathcal{H})}$ already includes the vertex appearance $(m,t_u^*)$.
			This implies that $t_v^* \notin W_j^u$ and furthermore, since $t_u,t_v \in W_i$ and $t_u - t_u^* \leq \Delta$ it follows that $t_v^* \in W_i$.
			Therefore, $um$ and $mv$ cannot be uncovered in $W_i$.
			
			What remains to show is that our modified solution is of the correct size.
			Namely $|\widetilde{\mathcal{C}(\mathcal{H})} | \leq |\mathcal{C}(\mathcal{H})| - k_v$, where $k_v$ is the number of vertex appearances of $v$ in $\mathcal{C}(\mathcal{H})$.
			We show this by proving that each vertex $(v,t_i) \in \mathcal{C}(\mathcal{H}) \cap \criticalset$ is disposable in the solution of $\mathcal{H}$.
			Remember, our process first inspects all the vertex appearances of form $(u,t_u) \in \mathcal{C}(\mathcal{H})$ and replaces them with the corresponding $(m,t_u^*)$.
			Since in the original solution $\mathcal{C}(\mathcal{H})$ the whole underlying path $P$ of $\mathcal{H}$ was covered, it follows that after performing the first two steps for the vertices $(u,t_u)$ the edge $um$ remains covered in $\mathcal{H}$.
			Even more, because each $(u,t_u)$ was exchanged by a corresponding appearance of the central vertex $m$ at time $t_u^*$ when both $um$ and $mv$ appear, it follows that the edge $mv$ is covered in $\mathcal{H}$.
			Therefore, we can denote all $(v,t_i) \in \mathcal{C}(\mathcal{H})$ as disposable in the solution of $\mathcal{H}$.

			Overall, this means that for every vertex appearance \( (v,t) \in \criticalset\) 
			and each of at least two Phase~1 \textsc{Partial $\Delta$-TVC} instances in \(\mathcal{H} \in \mathfrak{H}(v,t)\),
			we can find one disposable vertex appearance in \(\mathcal{C}(\mathcal{H})\).
			On the other hand, these found disposable vertex appearances coincide for at most two vertex appearances in \(\criticalset\) (this is true as $D$ can admit vertex appearances of both endpoints of the underlying path of a \textsc{Partial $\Delta$-TVC} instance $\mathcal{H}$).
			By Hall's theorem~\cite{Hall35} this means we can match each vertex appearance in \(\criticalset\) to a disposable vertex appearance in \(\mathcal{C}\), yielding that \(|\criticalset|\) is upper-bounded by \(\sum_{\mathcal{H} \in \mathfrak{H}} 
			\left(
			|\mathcal{C}(\mathcal{H})| - |\mathcal{X}(\mathcal{H})|\right)\) as desired.
		\end{proof}

		All of the above shows the desired approximation result.
		\begin{theorem}
			For every \(\Delta \geq 2\) and $d\geq 3$, \dTVC\ can be \((d-1)\)-approximated in time \(\mathcal{O}(|E(G)|^2T^2)\).
		\end{theorem}

		\subsection{An \FPT\ algorithm with respect to the solution size}\label{cp-alg-subsec}
		Our final result settles the complexity of \textsc{\(\Delta\)-TVC} from the viewpoint of parameterized complexity theory~\cite{CyganFKLMPPS15,DowneyFellows13,Niedermeier06} with respect to the standard parameterization of the size of an optimum solution.
		
		\begin{theorem}
			For every $\Delta\geq 2$, \dTVC\ can be solved in $\mathcal{O}((2\Delta)^kT n^2)$ time, where \(k\) is the size of an optimum solution.
			In particular \dTVC\ is in \FPT\ parameterized by \(k\).
		\end{theorem}
		\iflong
		\begin{proof}
			The algorithm starts by selecting an edge $e \in E(G)$ with an uncovered appearance in an arbitrary \(\Delta\)-time window $W_t$ and selects one of (at most) $\Delta$ vertex appearances of either endpoint to cover it.
			By covering one edge at some time point, we also cover some other edges in some other time windows,
			these are also excluded from the subsequent search.
			This step is repeated until after $k$-steps the algorithm stops and checks if selected vertices form a sliding $\Delta$-time window temporal vertex cover of $\mathcal{G}$.
			
			The algorithm builds at most $(2 \Delta)^k$ sets, for each of which checking if it is a temporal vertex cover takes $\mathcal{O}(T n^2)$ time.
			Therefore the running time of the algorithm is $\mathcal{O}((2 \Delta)^k T n^2)$.
		\end{proof}
		\fi
\section{Conclusion}
We presented a comprehensive analysis of the complexity of temporal vertex cover in small-degree graphs. 
We showed that \dTVC\ is NP-hard on the underlying paths (and cycles) already for $\Delta = 2$, while \textsc{TVC} on these underlying graphs can be solved in polynomial time. 
An open question that arises is for what value of $\Delta$ we have that \dTVC\ is NP-hard while $(\Delta + 1)$-\textsc{TVC} is polynomial-time solvable. This would determine the borders and achieve a complete dichotomy of the problem in terms of the value of $\Delta$.

Additionally, we provided a Polynomial-Time Approximation Scheme (PTAS) for \dTVC\ on temporal paths and cycles, complementing the hardness result. 
We present also an exact dynamic algorithm, that we then use in the $d-1$ approximation algorithm (where $d$ is the maximum vertex degree in any time step), which improves the $d$-approximation algorithm presented by Akrida \etal~\cite{AkridaMSZ20}.
Furthermore, we settled the complexity of \dTVC\ from the viewpoint of parameterized complexity theory, showing that \dTVC\ is fixed-parameter tractable when parameterized by the size of an optimum solution.

\bibliographystyle{abbrv}

\end{document}